\documentclass[a4paper]{ZbRad}
\usepackage{amssymb}
\usepackage{amscd}
\usepackage[dvipsnames]{pstricks}
\usepackage{pst-3dplot,amsmath, pst-poly}

\theoremstyle{plain}
 \newtheorem{thm}{Theorem}[section]
 \newtheorem{prop}{Proposition}[section]
 \newtheorem{lem}{Lemma}[section]
 
\theoremstyle{definition}
 
 \newtheorem{rem}{Remark}[section]
 \newtheorem{dfn}{Definition}[section]
\numberwithin{equation}{section}

\renewcommand{\ge}{\geqslant}
\renewcommand{\setminus}{\smallsetminus}
\DeclareMathOperator{\diag}{diag}
\allowdisplaybreaks

\newfont{\TITf}{cmssdc10 scaled 1440}

\title[Rigid body dynamics: classical and algebro-geometric integration]{}

\author[Borislav Gaji\'c]{}

\begin{document}

\setcounter{page}{1}

\vspace*{40mm}

\thispagestyle{empty}

{
\TITf\setlength{\parskip}{\smallskipamount}

\begin{center}
Borislav Gaji\'c

\bigskip\bigskip\bigskip

{THE RIGID BODY DYNAMICS: CLASSICAL AND ALGEBRO-GEOMETRIC INTEGRATION}

\end{center}
}

\vspace{4em}
{\leftskip3em\rightskip3em
\emph{Abstract}. The basic notion for a motion of a heavy rigid body fixed at a point
in three-dimensional space as well as its higher-dimensional generalizations are presented.
On a basis of the Lax representation, the algebro-geometric integration procedure
 for one
of the classical cases of motion of three-dimensional rigid body -
the Hess--Appel'rot system is given.
The classical integration in Hess coordinates is presented also.
For higher-dimensional generalizations, the special attention is
paid in dimension four. The L-A pairs and the classical integration procedures for completely integrable four-dimensional
rigid body so called the Lagrange bitop as well as for four-dimensional generalization of Hess--Appel'rot system are given.
An $n$-dimensional generalization of the Hess--Appel'rot system is
also presented and its Lax representation is given. Starting from another Lax representation
for the Hess--Appel'rot system, a family of dynamical systems on
$e(3)$ is constructed. For five cases from the family, the classical and algebro-geometric integration procedures are presented.
The four-dimensional generalizations for the Kirchhoff and the Chaplygin cases of motion of rigid body in ideal fluid are defined.
The results presented in the paper are part of results obtained in the last decade.


 \par}

\newpage\thispagestyle{empty}

\maketitle
\tableofcontents

\section{Introduction}

The rigid body motion is one of the most studied and most interesting systems of classical mechanics.
Nevertheless, there are still some important open questions and problems concerning it. In this paper we will focus on the problem of
integrability of motion of heavy rigid body fixed at a point. We are going to present the classical integration of some known integrable cases
as well as the algebro-geometric integrations procedure based on existence of the Lax representation.

One of the main questions in a study of the system of differential equations of motion of some mechanical system is
integrability or solvability. The notion of integrability is very close to the existence of the first integrals, i.e. functions that are constants
on solutions of the system. Until the beginning of XX century, the theory of integrable system had been intensively developed with
great influence of leading mathematicians and mechanicians of that time (Euler, Hamilton, Jacobi, Lagrange, Poincar\'e, Liouville, Noether, Kowalevski
an many others). For proving integrability they usually used some of the basic methods of that time: method of
separation of variables and Noether's theorem for finding integrals from symmetries. It became clear that algebraic geometry and theory of theta
functions, that was intensively developed in that time, have an important role in integration of the dynamical systems. For example,
the solutions of the Euler and Lagrange cases of motion of a rigid body fixed at a point are meromorphic functions on an elliptic curve. Starting from
that fact, Sofia Kowalevski formulated the problem of finding all cases of rigid body motion fixed at a point whose solutions are unique functions
of complex time that admit only moving poles as singularities. She proved that this is possible only in one more case, today called the Kowalevski case.
She found the additional first integral and she completely solved the system in theta functions. The importance of the Kowalevski paper is reflected
in the fact that thousands of papers are devoted to it. For recent progress, geometric interpretation and certain generalizations of
the Kowalevski top see \cite{Dr1, DK}.

In the 60's of XX century the big progress was made in the theory of integrable systems. It was proved that some nonlinear partial differential
equations (Korteveg-de Vries (KdV), Kadomtsev-Petviashvili (KP) and others) are infinitely-dimensional Hamiltonian systems. Also, a new method appeared:
 algebro-geometric integration procedure. It is based on the existence of a Lax representation (or L-A pair). A system admits L-A pair
with spectral parameter if there exist matrices $L(\lambda), A(\lambda)$ such that equations of the system can be written in the form:
\begin{equation}
\frac{d}{dt}L(\lambda)=\big[L(\lambda),A(\lambda)\big],
\label{la}
\end{equation}
where $\lambda$ is a complex number. An important case, when $L(\lambda)$ and $A(\lambda)$ are matrix polynomials in $\lambda$, was
studied by Dubrovin in \cite{D1} (see also \cite{D2, DKN}). The first consequence of \eqref{la} is that the spectrum of matrix $L(\lambda)$
is a constant function in time, i.e. coefficients in spectral polynomial are first integrals. If, from L-A pair, one gets enough integrals
for integrability, then the system can be integrated using algebro-geometric integration procedure, which is developed by the Novikov
school. In that procedure, the Baker--Akhiezer function plays the key role. This function is common eigenfunction of operators
$\frac{d}{dt}+A(\lambda)$ and $L(\lambda)$, defined on the spectral curve $\Gamma$ naturally associated to L-A pair. The Baker--Akhiezer function
is meromorphic on $\Gamma$ except in several isolated points where it has essential singularities. For a detailed
explanation see \cite{D1, D2, DKN, Dr, BBIM, Ga}. Let us mention also that the Lax representation is useful for constructing
higher-dimensional generalizations of the system.
In \cite{AvM1} Adler and van Moerbeke have presented an additional approach for integrability.
Both methods have been successfully applied to the rigid body motion (see \cite{M, BRS, RvM, R}).

The theories of rigid body motion and of integrable dynamical systems have been intensively studied
by Serbian scientists (see for example books and monographs
\cite{Bi, DR, DM, AS}).
At the Seminar Mathematical Methods of Mechanics in the Mathematical Institute SANU, supervised
by Vladimir Dragovi\' c, a group of young researchers including myself, Milena Radnovi\' c and Bo\v zidar Jovanovi\' c, the theory of integrable
dynamical systems has been studied for almost 20 years.
Here we will review some of the joint results obtained with Vladimir Dragovi\' c in the last decade (see \cite{DG, DG1, DG2, DG5, DG3}).

This paper is organized as follows. In Section 2 the notions of Poisson structure and integrability in Liouville sense are given. Also the basic
steps in algebro-geometric integration procedure are performed.
The basic facts about three-dimensional motion of a rigid body are presented in Section 3.
The classical as well as the algebro-geometric integration procedures for the Hess--Appel'rot case of motion of three-dimensional rigid body are given also.
The basic facts on higher-dimensional rigid body motion as well as the definition of the Lagrange bitop and $n$-dimensional Hess--Appel'rot
systems are presented in Section 4.
In Section 5 we present a construction of a class of systems on the Lie algebra $e(3)$.
For the five cases when an invariant measure is preserved, the classical and algebro-geometric integration procedures are given.
The four-dimensional generalizations of the Kirchhoff and Chaplygin cases of the motion of the rigid body in ideal fluid are given in
Section 6.

\section{Poisson structure on manifolds. Integrability. Algebro-geometric integration procedure}

Let $M$ be a smooth manifold, and $C^\infty(M)$ algebra of smooth functions on $M$.
\begin{dfn} A Poisson bracket on $M$ is a map $\{\,,\,\}:C^\infty(M)\times C^\infty(M)\to C^\infty(M)$
that for $f,g,h\in C^\infty(M)$ satisfies:
\begin{enumerate}
\item bilinearity: $\{\lambda f+\mu g,h\}=\lambda \{f,h\}+\mu\{g,h\},\quad \lambda,\mu\in\mathbb{R}$
\item skew-symmetry: $\{f,g\}=-\{g,f\}$
\item Leibnitz rule: $\{fg,h\}=f\{g,h\}+g\{f,h\}$
\item Jacobi identity: $\{f,\{g,h\}\}+\{g,\{h,f\}\}+\{h,\{f,g\}\}=0$.
\end{enumerate}
\end{dfn}

If $(x^1,x^2,...,x^n)$ are coordinates on $M$, using the Leibnitz rule one has:
$$
\{f,g\}=\sum_{i,j}\frac{\partial f}{\partial x^i}\frac{\partial g}{\partial x^j}\{x^i,x^j\}=
\sum_{i,j}P^{ij}\frac{\partial f}{\partial x^i}\frac{\partial g}{\partial x^j},
$$
where $P^{ij}=\{x^i,x^j\}$.
Poisson bracket is also called \emph{Poisson structure on manifold}, and a manifold endowed with Poisson structure is a \emph{Poisson manifold}.

Poisson bracket can be degenerate.
Then matrix $P=(P^{ij})$ is singular.
If $P$ is nonsingular, then, because of skew-symmetry, the dimension of $M$ is even and inverse matrix $P^{-1}$ gives symplectic structure on $M$.
Functions whose Poisson bracket with any other function is equal to zero are called the Casimir functions.

For a smooth function $H$ on manifold $M$, the system of equations
$$
\dot{x}^i=\{H,x_i\}=\sum_j P^{ij}\frac{\partial H}{\partial x^j}
$$
is called \emph{Hamiltonian system} with the Hamiltonian function $H$. The vector field $X_H^i=\sum_j P^{ij}\frac{\partial H}{\partial x^j}$ is called
the \emph{Hamiltonian vector field} associated with $H$.

A function $f$ is a first integral of a system of differential equations
$$
\dot{x}^i=X^i(x),\quad i=1,...,n
$$
if it is constant along every solution of the system, or in other words, if $\sum_i\frac{\partial f}{\partial x^i}X^i=X(f)=0$. Geometrically, it means
that each solution lies on a hypersurface $f=const.$
For the integrability in quadratures one needs $n-1$ first integrals. However, by the Jacobi theorem, if a system preserves the standard measure, i.e.
if the divergence of the vector field $X$ is zero, then for the integrability in quadratures one needs only $n-2$ first integrals (see \cite{AKN,G}).

For Hamiltonian systems there is an additional structure, the Poisson structure.
A function $f$ is a first integral of a Hamiltonian system with the Hamiltonian function $H$ if it Poisson-commutes with the Hamiltonian,
i.e. if $\{f,H\}=0$.
The following theorem is fundamental concerning the integrability of the Hamiltonian systems.

\begin{thm}[Liouville-Arnol'd] Let $M^{2n}$ be a symplectic manifold, and $f_1=H,f_2,...,f_n$ functions that satisfy $\{f_i,f_j\}=0$.
Denote $M_f=\{x\in M| f_1(x)=const,..., f_n(x)=const.\}$.
If $f_1,...,f_n$ are functionally independent on $M_f$ then
\begin{enumerate}
\item $M_f$ is a smooth manifold invariant under the Hamiltonian flow with the Hamiltonian $H=f_1$.
\item If $M_f$ is compact and connected, then it is diffeomorphic to $n$-dimensional torus $T^n$.
\item There exist coordinates $\varphi=(\varphi_1,...,\varphi_n)$ on $T^n$ in which the Hamiltonian flow is linearized:
$$
\dot{\varphi_i}=\omega_i(f_1,...,f_n).
$$
\item The Hamiltonian system with Hamiltonian $H$ can be solved in quadratures.
\end{enumerate}
\end{thm}

For proof see \cite{Ar}.

The Hamiltonian system that satisfies the Arnol'd-Liouville theorem will be called \emph{completely integrable in the Liouville sense}.

\subsection{The basic steps of the algebro-geometric integration procedure}

We will give here a short description of algebro-geometric integration procedure. For details see \cite{D1, D2, DKN, DMN, Dr, BBIM}.

The existence of a Lax representation \eqref{la} for a system of differential equations is equivalent to commutativity of operators:
$$
\Big[\frac{d}{dt}+A(t,\lambda),L(t,\lambda)\Big]=0.
$$
Let $\Phi(t,\lambda)$ is the fundamental solution matrix for the equation:
\begin{equation}
\Big(\frac{d}{dt}+A(t,\lambda)\Big)\Phi(t,\lambda)=0,
\label{la1}
\end{equation}
normalized by the condition $\Phi(0,\lambda)=1$.
From the Lax representation one gets that $L(t,\lambda)\Phi(t,\lambda)$ is also a solution of \eqref{la1}.
Since every solution is determined by its initial conditions, we have
$$
L(t,\lambda)\Phi(t,\lambda)=\Phi(t,\lambda)L(0,\lambda).
$$
Consequently, the matrices $L(t,\lambda)$ and $L(0,\lambda)$ have the same spectrum. In other words, the coefficients of the
characteristic equation
\begin{equation}
p(\lambda,\mu)=\det(L(\lambda)-\mu\cdot1)=0
\label{sk}
\end{equation}
are the first integrals of the system.

The equation \eqref{sk} defines algebraic curve $\Gamma$, called \emph{the spectral curve}. Let $L$ and $A$ are $n\times n$ matrices.
Over $\lambda \in \mathbb{C}$ we have $n$ points
on $\Gamma$ with coordinates $(\lambda, \mu_1),...,(\lambda,\mu_n)$. To each of these points corresponds eigenvector $h(t,(\lambda,\mu))$ of
matrix $L(t,\lambda)$:
$$
L(t,\lambda)h(t,(\lambda,\mu_k))=\mu_kh(t,(\lambda,\mu_k)).
$$
Fix the normalization
$$
\sum_{i=1}^n h_i(t,(\lambda,\mu))=1.
$$
Normalized vector $h$ can be regarded as meromorphic vector-function on $\Gamma$. Introduce the function
$$
\psi(t,(\lambda,\mu))=\Phi(t,\lambda)h(0,(\lambda,\mu)).
$$
In what follows we will see that this function has a key role in the integration procedure.

On can easily check that $\psi(t,(\lambda,\mu))$ satisfies the following relations:
$$
\begin{aligned}
L(t,\lambda)\psi(t,(\lambda,\mu_k))&=\mu_k\psi(t,(\lambda,\mu_k)),\\
\Big(\frac{d}{dt}+A(t,\lambda)\Big)\psi(t,(\lambda,\mu_k))&=0.
\end{aligned}
$$
The following theorem gives us the analytical properties of $\psi(t,(\lambda,\mu))$. Denote with $P_1,...,P_n$ the points over
$\lambda=\infty$.

\begin{thm}\label{ba}\cite{DKN} The vector-function $\psi(t,(\lambda,\mu))$ has the following properties:
\begin{enumerate}
\item It is meromophic on $\Gamma\setminus\{P_1,...,P_n\}$. Its divisor of poles has degree $g+n-1$, and it does not depend on time, where $g$
is the genus of the curve $\Gamma$.
\item In the neighborhood of $P_k$ the function $\psi(t,(\lambda,\mu))$ has the form:
$$
\psi(t,(\lambda,\mu))=\alpha(t,z_k^{-1})\exp[q_k(t,z_k)],
$$
where $z_k^{-1}$ is a local coordinate in the neighborhood of $P_k$, $\alpha$ is the holomorphic vector-function, and $q_k$ are polynomials.
\end{enumerate}
\end{thm}

 The functions that satisfy conditions from theorem \ref{ba} are called \emph{$n$-point Baker--Akhiezer functions}.
The example of such function is the exponential function $e^{1/z}$ on $\mathbb{CP}^1$. It is holomoprphic everywhere except in point $z=0$,
where it has an essential singularity. The most general definition of the Baker--Akhiezer functions is given by Krichever. For a history and details see
\cite{DKN}.
\begin{dfn}\label{bad}(see \cite{DKN}) Let $P_1,...,P_n$ are points on a Riemann surface $\Gamma$ of genus $g$ and $z_k^{-1}$ are local coordinates in the
neighborhoods of these points such that $z_k^{-1}(P_k)=0$. Let $q_1(z),...,q_k(z)$ are polynomials and let $\mathcal{D}$ is divisor on $\Gamma$.
The $n$-point Baker--Akhiezer function $\psi$ is a function that satisfies the following conditions: it is meromorphic on
$\Gamma\setminus\{P_1,...,P_n\}$, its divisor of zeros and poles satisfies $(\psi)+\mathcal{D}\ge0$, and in the neighborhood of the each point $P_k$
the function $\psi(P)\exp(-q_k(z_k(P)))$ is analytic.
\end{dfn}
The basic idea of algebro-geometric integration procedure is to reconstruct $\psi(t,(\lambda,\mu))$. From
the analytical properties given in theorem \ref{ba}, one can in terms of theta functions explicitly construct $\psi(t,(\lambda,\mu))$,
and using it solve the system.

\begin{thm} For a non-special divisor $\mathcal{D}$ of degree $N$ the dimension of the linear space of functions with properties from definition \ref{bad}
is equal $N-g+1$. Particularly, if the degree of $\mathcal{D}$ is $g$, then $\psi$ is uniquely determined up to factor and it is given by:
$$
\psi(P)=c\ \exp\Big(\sum_{k=1}^n\int_{P_0}^P\Omega_{q_k}\Big)
\frac{\theta\big(\mathcal{A(P)}+\sum_kU^{(q_k)}-\mathcal{A(\mathcal{D})}-K\big)}
{\theta\big(\mathcal{A(P)}-\mathcal{A(\mathcal{D})}-K\big)}
$$
where  $\Omega_{q_k}$ are normalized Abelian differentials of the second order, which in the neighborhood of $P_k$ has the form
$$
\Omega_{q_k}=dq_k(z_k(P))+\text{holomorphic part},
$$
the vector $U^{(q_k)}$ is a vector of $b$-periods of the differential $\Omega_{q_k}$, and $K$ is the vector of Riemann constants.
\end{thm}

As a corollary, in a generic situation, one has $n$ functions $\psi^1,...,\psi^n$ that satisfy the theorem \ref{ba}. Let $\Psi(t,\lambda)$
be the matrix of which the columns are vectors $\psi(t,(\lambda,\mu_k)), k=1,...,n$ (here $\mu_k$ are eigenvalues of $L(\lambda)$). One has
$$
L(t,\lambda)=\Psi\hat\mu\Psi^{-1},\quad A(\lambda)=\frac{\partial\Psi}{\partial t}\Psi^{-1},
$$
where $\hat\mu=\diag(\mu_1,...,\mu_n)$. Consequently, in a general situation, from the Baker-Akhiezer functions one can find the
matrix $L$ and $A$ as functions of time, or in other words one can integrate the system.

\section{Motion of a heavy rigid body fixed at a point}\label{s3}

A three-dimensional rigid body is a system of material points in $\mathbb{R}^3$ such that the distance between each two points
is a constant function of time.
We shall consider motion of a rigid body with a fixed point $O$.
The configuration manifold is the Lie group $SO(3)$.
Two different Euclidian coordinate frames are associated to the system: the first one $Oxyz$ is fixed in the space, and the second, moving, $OXYZ$
is fixed in the body.
With the capital letters we will denote elements of the moving reference frame, while the lowercase letters will denote elements of the
fixed reference frame. If some point of the body has the radius vector $\vec{Q}$ in the moving coordinate system, then its radius vector in the fixed frame
is $\vec{q}(t)=B(t)\vec{Q}$, where $B(t)\in SO(3)$ is an orthogonal matrix. The velocity of that point in the fixed reference frame is given by
$$
\vec{v}(t)=\dot{\vec{q}}(t)=\dot{B}(t)\vec{Q}=\dot{B}(t)B^{-1}(t)\vec{q}(t)=\omega(t)\vec{q}(t),
$$
where $\omega(t)=\dot{B}B^{-1}$. It can be proved that $\omega$ is an skew-symmetric matrix.
Using the isomorphism of $(\mathbb{R}^3, \times)$, where $\times$ is the usual vector product, and
 $(so(3), [\,,\,])$, given by
\begin{equation}
\vec{a}=(a_1,a_2,a_3)\mapsto a=\left[\begin{matrix} 0&-a_3&a_2\\
                                                    a_3&0&-a_1\\
                                                    -a_2&a_1&0
                                                    \end{matrix}\right]
\label{so3}
\end{equation}
matrix $\omega(t)$ is corresponded to vector $\vec{\omega}(t)$ - \emph{angular velocity of the body} in the fixed reference frame.
Then $\vec{v}(t)=\vec{\omega}(t)\times\vec{q}(t)$.
One can easily see that $\vec{\omega}(t)$ is the eigenvector of matrix $\omega(t)$ that corresponds to the zero eigenvalue.

In the moving reference frame, $\vec{V}(t)=B(t)^{-1}\vec{v}(t)$, so $\vec{V}(t)=\vec{\Omega}(t)\times\vec{Q}$,
where $\vec{\Omega}(t)$ is the angular velocity in the moving reference frame and corresponds to the skew-symmetric matrix
$\Omega(t)=B^{-1}(t)\dot{B}(t)$.

At the fixed moment of time, the vector $\vec\omega$ defines the line $l$ through the fixed point. For points on $l$,
vectors $\vec q(t)$ and $\vec\omega(t)$ are collinear. Hence, velocities of points on $l$ are equal to zero. The line $l$ is called
\emph{instantaneous axis of rotation}.

The existence of instantaneous axis of rotation can be regarded as a infinitesimal version of Euler's rotation theorem. The theorem states that
any finite displacement of a rigid body with fixed point is equivalent to the rotation about some axes through the fixed point.
Here we will present the original proof given by Euler in 1775 \cite{Eu}.

\begin{thm}[Euler \cite{Eu}] If a sphere is turned about its centre, it is always possible to assign a diameter,
whose direction in the displaced position is the same as in the initial position.
\end{thm}

\begin{proof}
Consider a great circle $k_1$ in initial state, which after
the displacement goes to a circle $k_2$.  Denote by
$A$ a point of the intersection of these two circles. Since $A$ belongs to $k_1$, after rotation it goes to the point $a\in k_2$.
On the other hand, the point $A$ is on $k_2$, so there is point $\alpha\in k_1$ which after rotation goes to $A$.
We will show that there exists a point $O_1$ on the sphere, which equally refers to the circle $k_1$ as to the
circle $k_2$. If we suppose that $O_1$ is constructed, then the arcs $O_1A$ and $O_1a$ should
be equal to each other. Also, the arcs $O_1A$ and $O_1a$ are similarly inclined towards the
circles $k_1, k_2$. 
Consequently, the angles $O_1aA$ and $O_1A\alpha$ are equal also.
But since the arcs $O_1a$ and $O_1A$ are equal, the angles $O_1aA$ and $O_1Aa$ are also equal, whence $O_1Aa=O_1A\alpha$.
It is clear that $O_1$ lies on the arc bisecting the angle $\alpha Aa$. So, $O_1$ can be constructed as an intersecting
point of the sphere and the following two planes through the fixed point $O$ (the center of the sphere). The first plane is the
symmetry plane of the angle $aA\alpha$ and the second one is the symmetry plane of the arc $Aa$. The diameter that we are looking for
is determined by the point $O_1$ and the center of the sphere.
\end{proof}

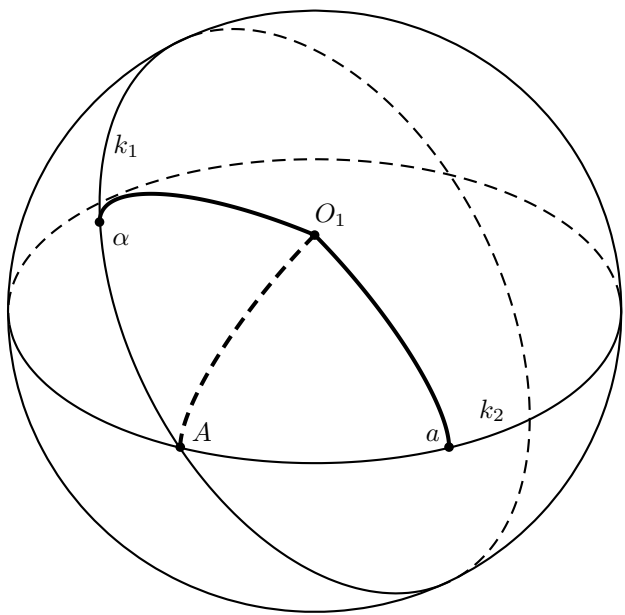
\def\radius{4 }\def\PhiI{20 }\def\PhiII{50 }
\begin{pspicture}(-6,-5)(4,5)
  \psset{Alpha=45,Beta=30,linestyle=solid}
\pstThreeDEllipse[linestyle=dashed,linecolor=black,beginAngle=-55, endAngle=125](0,0,0)(0,\radius,0)(0,0,\radius)
\pstThreeDEllipse[linecolor=black,beginAngle=125, endAngle=305](0,0,0)(0,\radius,0)(0,0,\radius)
  \pstThreeDEllipse[linecolor=black, beginAngle=-40, endAngle=140](0,0,0)(4.03,0,0)(0,4.03,0)
  \pstThreeDEllipse[linestyle=dashed, linecolor=black, beginAngle=140, endAngle=320](0,0,0)(4.03,0,0)(0,4.03,0)
\pstThreeDCircle[linecolor=black](0,0,0)(-1.78,-1.78,3.15)(-3.3,3.3,0)

\pstThreeDEllipse[linestyle=solid,linewidth=1.5pt,linecolor=black,beginAngle=0, endAngle=90](0,0,0)(1.3,3.8,0)(2,2,2.8)
\pstThreeDEllipse[linestyle=solid,linewidth=1.5pt,linecolor=black,beginAngle=0, endAngle=90](0,0,0)(4,0,3)(2,2,2.8)
\pstThreeDEllipse[linestyle=dashed,linewidth=1.5pt,linecolor=black,beginAngle=0, endAngle=90](0,0,0)(3.8,1.3,0)(2,2,2.8)

\pstThreeDDot (3.8,1.3,0)
\pstThreeDPut(3.8,1.7,0.4) {$A$}

\pstThreeDDot (1.3,3.8,0)
\pstThreeDPut(1.6,3.8,0.3) {$a$}

\pstThreeDDot (4,0,3)
\pstThreeDPut(4,0.4,2.9) {$\alpha$}

\pstThreeDDot (2,2,2.8)
\pstThreeDPut(0,0.3,1.6) {$O_1$}

\pstThreeDPut(3.8,0.3,4.2) {$k_1$}
\pstThreeDPut(-0.4, 2.9,-0.5) {$k_2$}



\end{pspicture}

Let us stress that it is natural to consider the angular velocity as an skew-symmetric matrix. The element $\omega_{12}$
corresponds to the rotation in the plane determined by the first two axes $Ox$ and $Oy$, and similarly for the other elements. In the
three-dimensional case we have a natural correspondence given above, and one can consider the angular velocity as a vector. But, in
higher-dimensional cases, generally speaking, such a correspondence does not exist. We will see later how in dimension four, using isomorphism
between $so(4)$ and $so(3)\times so(3)$ two vectors in the three-dimensional space are joined to an $4\times 4$ skew-symmetric matrix.
Since we cannot imagine higher-dimensional world (or at least, it is not easy to imagine it), it is much easier to consider,
for example, two dimensional world. If the two-dimensional people consider rotation of a rigid body with fixed point, they conclude that angular
velocity is a two-dimensional skew-symmetric matrix. If the element $\omega_{12}$ is positive, then rotation goes in the positive sense, otherwise
it goes in the negative one. There is no third dimension, so they cannot conclude that it can be seen as a rotation about $z$-axis.
In the two-dimensional world the $z$-axis does not exist!

\emph{The moment of inertia} with respect to the axis $u$, defined with the unit vector $\vec{u}$ through a fixed point $O$ is :
$$
I(u)=\sum m_id_i^2=\sum m_i\langle\vec{u}\times\vec{Q}_i,\vec{u}\times\vec{Q}_i\rangle=\Big\langle\sum m_i\vec{Q}_i\times
(\vec{u}\times\vec{Q}_i), \vec{u}\Big\rangle=\langle I\vec{u},\vec{u}\rangle,
$$
where $d_i$ is the distance between $i$-th point and axis $u$,
and $I$ is inertia operator with respect to the point $O$ defined with
$$
I\vec u=\sum m_i\vec{Q}_i\times(\vec{u}\times\vec{Q}_i).
$$
In coordinates $(X,Y,Z)$ the diagonal
elements $I_{11}, I_{22}, I_{33}$ of $I$ are moments of inertia of the body with respect to the
coordinate axes $OX, OY, OZ$ respectively. For example $I_{11}=\sum m_i(Y_i^2+Z_i^2)$. Non-diagonal elements
are called \emph{centrifugal moments of inertia}. For example $I_{12}=-\sum m_i X_i Y_i$ and similar for other $I_{ij}$.
One can easily see that $I$ is symmetric and positive definite operator and consequently, one can choose an orthogonal basis in which
the operator has the diagonal form $I=\diag(I_1, I_2, I_3)$. Then  $I_1, I_2, I_3$ are called \emph{the principal
moments of inertia}, with respect to the \emph{principal axes of inertia}. If some of $I_1, I_2, I_3$ coincide,
for example if $I_1=I_2$, then any axis in the coordinate plane $XOY$ is principal. The ellipsoid $\langle I\Omega,\Omega\rangle=1$ is called
\emph{inertia ellipsoid of the body} at a point $O$. In the principal coordinates its equation is:
$$
I_1\Omega_1^2+I_2\Omega_2^2+I_3\Omega_3^2=1.
$$
Any symmetry of the body gives the symmetries for the  inertia ellipsoid. For example the regular hexagon with homogeneous mass distribution
is invariant under rotations by $\pi/3$ around the normal line throw the center. Consequently, $I_1=I_2$, and any axis in the plane of the hexagon
through its center is the principal axis. The similar conclusion can be derived for the star (see picture 2).
So, we have here two geometrically different objects with the same inertia momenta.

\begin{pspicture}(-3,-1)(1,3)
\PstHexagon\hskip1cm
\PstPolygon[PolyIntermediatePoint =0.38]
\end{pspicture}

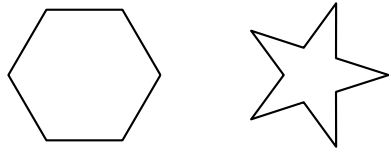
 Picture 2: The regular hexagon and the star have $I_1=I_2$
\medskip

The kinetic energy of the body is given by:
$$
\begin{aligned}
T&=\frac12\sum m_iV_i^2=\frac12\sum m_i\langle\vec{\Omega}\times\vec{Q}_i, \vec{\Omega}\times\vec{Q}_i\rangle=\\
&=\frac12\langle\vec{\Omega}, \sum m_i\vec{Q}_i\times(\vec{\Omega}\times\vec{Q}_i)\rangle=\frac12\langle I\vec{\Omega}, \vec{\Omega}\rangle
\end{aligned}
$$

Similarly, for the angular momentum $\vec{M}$ with respect to the point $O$, we have:
$$
\vec{M}=\sum \vec{Q}_i\times (m_i\vec{V}_i)=\sum m_i\vec{Q}_i\times(\vec{\Omega}\times\vec{Q}_i)=I\vec{\Omega}.
$$

We consider a motion of a heavy rigid body fixed at a point. Let us denote by $\vec{\chi}$ the radius vector of the center of
masses of the body multiplied with mass $m$ of the body and gravitational acceleration $g$. By $\vec{\Gamma}$ we denote
the unit vertical vector.

The motion in the moving reference frame is described by the Euler-Poisson equations \cite{G,BM}:
\begin{equation}
\begin{aligned}
\dot{\vec{M}}&=\vec{M}\times\vec{\Omega}+\vec{\Gamma}\times\vec{\chi}\\
\dot{\vec{\Gamma}}&=\vec{\Gamma}\times\vec{\Omega}.
\end{aligned}
\label{EP}
\end{equation}
Using that $\vec{M}=I\vec{\Omega}$, one see that \eqref{EP} as a system of six ordinary differential equations in $\vec{M}$ and
$\vec{\Gamma}$ with six parameters $I=\diag(I_1,I_2,I_3)$, $\vec{\chi}=(X_0,Y_0,Z_0)$.
These equations have three first integrals:
\begin{equation}
\begin{aligned}
H&=\frac{1}{2}\langle\vec{M},\vec{\Omega}\rangle+\langle\vec{\Gamma},\vec{\chi}\rangle\\
F_1&=\langle\vec{M},\vec{\Gamma}\rangle,\\
F_2&=\langle\vec{\Gamma},\vec{\Gamma}\rangle.
\end{aligned}
\label{PI}
\end{equation}

Since the equations preserve the standard measure, by Jacobi theorem (see for example \cite{G,AKN}) for integrability in
quadratures one needs one more additional functionally independent first integral.

On the other hand, the equations \eqref{EP} are Hamiltonian on the Lie algebra $e(3)$ with the standard Lie-Poisson structure:
\begin{equation}
\{M_i, M_j\}=-\epsilon_{ijk} M_k,\ \ \{M_i,
\Gamma_j\}=-\epsilon_{ijk} \Gamma_k,\ i,j,k=1,2,3.
\label{PS}
\end{equation}
The structure \eqref{PS} has two Casimir functions $F_1$ and $F_2$ from \eqref{PI}. So, symplectic leafs
are four-dimensional (they are diffeomorphic to the cotangent bundle of the two-dimensional sphere \cite{Koz1})
and for the integrability in Liouville sense one needs, besides the Hamiltonian $H$ from \eqref{PI}, one more functionally independent first integral.

From the facts given above, one concludes that a natural problem arises: for which values of the parameters $I_1, I_2, I_3, X_0, Y_0, Z_0$,
the equations \eqref{EP} admit the fourth functionally  independent first integral?

\subsection{Integrable cases}

Existence of additional  independent fourth integral gives strong restrictions on moments of inertia and  vector $\vec{\chi}$. Such integral
exists only in three cases:

\begin{itemize}
\item Euler case (1758): $X_0=Y_0=Z_0=0$. The additional integral is $F_4=\langle M, M\rangle$.

\item Lagrange case (1788): $I_1=I_2$, $\vec{\chi}=(0,0,Z_0)$. The additional integral is $F_4=M_3$.

\item Kovalewski case (1889): $I_1=I_2=2I_3$, $\vec{\chi}=(X_0,0,0)$. The additional integral is
$F_4=(\Omega_1^2-\Omega_2^2+\frac{X_0}{I_3}\Gamma_1)^2+(2\Omega_1\Omega_2+\frac{X_0}{I_3}\Gamma_2)^2$
\end{itemize}

We have also cases that admits a fourth integral only with a fixed value of one of the integrals. If Casimir function $F_1=0$,
then we have
\begin{itemize}
\item Goryachev-Chaplygin case (1900): $I_1=I_2=4I_3,\ \vec\chi=(X_0,0,0)$. The additional integral is $F_4=M_3(M_1^2+M_2^2)+2M_1\Gamma_3$;
\end{itemize}

Beside the completely integrable cases, there are cases that instead of additional first integral have an invariant relation. We will focus on
Hess--Appel'rot case.
Hess in \cite{He} and Appel'rot in \cite{Ap} found
that if the inertia momenta and the radius vector of the centre of
masses satisfy the conditions
\begin{equation}
Y_0=0,\qquad X_0\sqrt{I_1(I_2-I_3)}+Z_0\sqrt{I_3(I_1-I_2)}=0,
\label{ha}
\end{equation}
then the surface
\begin{equation}
F_4=M_1X_0+M_3Z_0=0
\label{4}
\end{equation}
is invariant. It means that if at the initial moment one has that
$F_4=0$, then this will be satisfied during the time evolution of the
system.

\subsection{Classical integration of Hess--Appel'rot case}

Classical integration of the Hess--Appel'rot system is done in the so-called Hess coordinates.
The details and historical notes can be found in \cite{G}.
Hess introduced new coordinates $\nu, \mu, \rho$:
\begin{equation}
\begin{aligned}
\nu&=M_1^2+M_2^2+M_3^2\\
\mu&=-(X_0\Gamma_1+Y_0\Gamma_2+Z_0\Gamma_3)+\frac{h}2=
\frac12\Big(\frac{M_1^2}{I_1}+\frac{M_2^2}{I_2}+\frac{M_3^2}{I_3}\Big),\\
\rho&=M_1X_0+M_2Y_0+M_3Z_0.
\end{aligned}
\label{hk}
\end{equation}
If one denotes
$$
\begin{aligned}
\tau&=\frac{M_1^2}{I_1^2}+\frac{M_2^2}{I_2^2}+\frac{M_3^2}{I_3^2},\\
\sigma&=\frac{M_1}{I_1}X_0+\frac{M_2}{I_2}Y_0+\frac{M_3}{I_3}Z_0,\\
\mu_1&=2\mu,\quad \delta^2=X_0^2+Y_0^2+Z_0^2,
\end{aligned}
$$
the equations of motion become (see \cite{G}):
$$
\begin{aligned}
\left(\frac 12\frac{d\nu}{dt}\right)^2&=\left|\begin{matrix}\delta^2&\mu&\rho\\
                \mu&1&c_1\\
                \rho&c_1&\nu
                \end{matrix} \right|,\qquad
\left(\frac{d\rho}{dt}\right)^2=\left|\begin{matrix}\delta^2&\rho&\sigma\\
                \rho&\nu&\mu_1\\
                \sigma&\mu_1&\tau
                \end{matrix} \right|,\\
(\nu\delta^2-\rho^2)\frac{d\mu}{dt}&=(\mu_1\delta^2-\rho\sigma)\frac12
\frac{d\nu}{dt}+(c_1\delta^2-\mu\rho)\frac{d\rho}{dt},
\end{aligned}
$$
where $h$ and $c_1$ are fixed values of the first integrals \eqref{PI}: $F_1=c_1,\, H=h$.

The equation of the invariant surface becomes $\rho=0$, and we get:
$$
\left(\frac 12\frac{d\nu}{dt}\right)^2=\left|\begin{matrix}\delta^2&\mu&0\\
                \mu&1&c_1\\
                0&c_1&\nu
                \end{matrix} \right|,\qquad
\nu\frac{d\mu}{dt}=\mu_1\frac12\frac{d\nu}{dt}.
$$
From the second equation one has $\mu=c\nu$, where $c$ is a constant, and from the first we have
$$
\frac{d\nu}{\sqrt{-c_1^2\delta^2+\delta^2\nu-c^2\nu^3}}=2dt.
$$
Hence, $\nu$ is $\nu=\Phi(t+\tilde{c})$, where $\Phi(t)$ is an elliptic function.

So, in the Hess--Appel'rot case, in the Hess coordinates one can find solutions:
$$
\mu=c\Phi(t+\tilde{c}),\quad\nu=\Phi(t+\tilde{c}),\quad\rho=0.
$$
Nevertheless, when Hess--Appel'rot conditions are satisfied, the coordinate transformation
\eqref{hk} becomes degenerate. Namely, from the system:
$$
\aligned
M_1^2+M_2^2+M_3^2&=\nu\\
\frac{M_1^2}{I_1}+\frac{M_2^2}{I_2}+\frac{M_3^2}{I_3}&=\mu_1\\
M_1^2X_0^2-M_3^2Z_0^2&=0
\endaligned
$$
one needs to find $M_1^2, M_2^2, M_3^2$ as functions of $\mu_1$ and $\nu$.
The determinant of the system is:
$$
\Delta=\frac{1}{I_1I_2I_3}(X_0^2I_1(I_2-I_3)-Z_0^2I_3(I_1-I_2))=0,
$$
hence, the coordinate transformation \eqref{hk} becomes degenerative.
Consequently, to find a solution of the system one needs to solve one more differential
equation. Hess proved in \cite{He} that additional differential equations can be reduced to a Riccati equation.
Let us present Hess result.

From conditions \eqref{4}, one gets $M_1=Z_0u,\quad M_3=-X_0u$.
Let $v$ is defined with $M_2=\delta v, \quad \delta=X_0^2+Z_0^2$.
From the first integrals
$$
\begin{aligned}
\Gamma_1^2+\Gamma_2^2+\Gamma_3^2&=1,\\
X_0\Gamma_1+Z_0\Gamma_3&=-\frac12\Big(\frac{M_1^2}{I_1}+\frac{M_2^2}{I_2}+\frac{M_3^2}{I_3}\Big)+
\frac{h}{2}\\
M_1\Gamma_1+M_2\Gamma_2&+M_3\Gamma_3=c_1=\beta
\end{aligned}
$$
by expressing $\Gamma_1,\Gamma_2,\Gamma_3$ and putting them in the first two equations of motion
$$
\begin{aligned}
\dot{M}_1&=\Big(\frac{1}{I_3}-\frac{1}{I_2}\Big)M_2M_3+\Gamma_2Z_0\\
\dot{M}_2&=\Big(\frac{1}{I_1}-\frac{1}{I_3}\Big)M_1M_3+\Gamma_3X_0-\Gamma_1Z_0
\end{aligned}
$$
one gets
$$
\begin{aligned}
\beta &Z_0\frac{udv+vdu}{u^2+v^2}-Z_0\sqrt {H_1}\frac{vdu-udv}{\delta(u^2+v^2)}+\\
X_0&Z_0^2\left(\frac1{I_3}-\frac1{I_1}\right)u^2du-
\delta^2X_0\left(\frac1{I_3}-\frac1{I_2}\right)uvdv=0,\\
H_1&=\delta^2\left[(u^2+v^2)\left[\delta^2-\left(\frac{(I_1-I_3)}{I_3(I_1-I_2)}
X_0^2(u^2+v^2)-h\right)^2\frac14\right]-\beta^2\right].
\end{aligned}
$$
The polynomial $H_1$ is of degree three in $\rho_1^2=u^2+v^2$,
$H_1=P_3(\rho_1^2)$. Introducing  $u=\rho_1\cos\varphi, v=\rho_1\sin\varphi$,
the previous equation reduces to
$$
\frac1{\delta}\sqrt{P_3(\rho_1^2)}d\varphi+\left(\frac{\beta}{2\rho_1}+
L\rho_1^2\cos\varphi\right)d\rho_1=0,\ \ L=2X_0Z_0\frac{I_1-I_3}{I_1I_3}.
$$
This is the equation derived by Hess. Nekrasov proved that this equation can be reduced
to a second order linear differential equation with double-periodical coefficients.
Introducing  $\tau=tg\frac\varphi2$,
the last equation becomes a Riccati equation:
$$
\frac{d\tau}{d\psi}=-\tau^2+\Phi(\rho_1)
$$
where
$d\psi=\frac{\delta\rho_1 d\rho_1}{\sqrt{P_3(\rho_1^2)}}$.
Introducing $\tau=\frac{ds/d\psi}{s}$ we have
$$
\frac {d^2s}{d\psi^2}=s\Phi(\rho_1(\psi))
$$
which is the equation obtained by Nekrasov.

\subsection{Lax representation for the classical Hess--Appel'rot system. Algebro-geometric integration}

The Lax representation for classical Hess--Appel'rot system, with the algebro-geometric integration procedure was presented in \cite{DG}.
It is proved there that the integration also leads to an elliptic function and an additional
Riccati differential equation.

Using isomorphism \eqref{so3}, equations \eqref{EP} can be written in the matrix form:
$$
\begin{aligned}
\dot{M}&=[M,\Omega]+[\Gamma,\chi]\\
\dot{\Gamma}&=[\Gamma,\Omega],
\end{aligned}
$$
where the skew-symmetric matrices represent vectors denoted by the
same letter.

We have the following:
\begin{thm}\cite{DG}
If condition \eqref{4} is satisfied, the equations of Hess--Appel'rot case can be written in the form:
\begin{equation}
\begin{aligned}
\dot{L}(\lambda)&=[L(\lambda), A(\lambda)],\\
L(\lambda)=\lambda^2 C+ \lambda M+&\Gamma,\quad A(\lambda)=\lambda\chi +\Omega, \quad C=I_2\chi.
\end{aligned}
\label{5}
\end{equation}

\end{thm}
The spectral curve, defined by:
$$
\mathcal{C}: \ \ p(\mu,\lambda):=\det(L(\lambda)-\mu E)=0,
$$
is:
$$
\mathcal{C}: \ \ -\mu(\mu^2-\omega^2+2\Delta\Delta^*)=0
$$
where
\begin{equation}
\begin{aligned}
\alpha&=\frac{X_0}{\sqrt{X_0^2+Z_0^2}}\quad\beta=\frac{Z_0}{\sqrt{X_0^2+Z_0^2}}\\
\Delta&=y+\lambda x,\qquad \Delta^*=\bar{y}+\lambda\bar{x},\\
y&=\frac1{\sqrt{2}}(\beta\Gamma_1-\alpha\Gamma_3-i\Gamma_2),\quad
x=\frac1{\sqrt{2}}(\beta M_1-\alpha M_3-iM_2),\\
\omega&=-i\left[\alpha(C_1\lambda^2+M_1\lambda+\Gamma_1)+
\beta(C_3\lambda^2+M_3\lambda+\Gamma_3)\right]\\
&=-i\left[\alpha(C_1\lambda^2+\Gamma_1)+
\beta(C_3\lambda^2+\Gamma_3)\right].
\end{aligned}
\label{5.05}
\end{equation}
This curve is reducible. It consists of two components:
the rational curve $\mathcal{C}_1$ given by $\mu=0$, and the elliptic curve
$\mathcal{C}_2$ :
\begin{equation}
\mu^2=P_4(\lambda)=\omega^2-2\Delta\Delta^*.
\label{5.1}
\end{equation}

The coefficients of the spectral polynomial  are integrals of motion.
If one rewrites the equation of the spectral curve in the form:
$$
p(\mu, \lambda)=-\mu(\mu^2+A\lambda^4+B\lambda^3+D\lambda^2+E\lambda+F)=0,
$$
one gets:
$$
\begin{aligned}
A&=I_2^2(X_0^2+Z_0^2),\\
B&=2I_2(M_1X_0+M_3Z_0)(=0),\\
D&=M_1^2+M_2^2+M_3^2+
2I_2(X_0\Gamma_1+Z_0\Gamma_3),\\
E&=2(M_1\Gamma_1+M_2\Gamma_2+M_3\Gamma_3),\\
F&=\Gamma_1^2+\Gamma_2^2+\Gamma_3^2(=1).
\end{aligned}
$$
So, L-A pair \eqref{5} gives three integrals and one invariant relation.

Here, we will review some basic steps in the algebro-geometric integration procdure from \cite{DG}.

Let $(f_1, f_2, f_3)^T$ denote an eigenvector of the matrix $L(\lambda)$,
which corresponds to the eigenvalue $\mu$. Fix normalizing condition
$f_1=1$.

Then one can prove:

\begin{lem}\cite{DG} The divisors of $f_2$ and $f_3$ on $\mathcal{C}_2$ are:
$$
\aligned
(f_2)&=-P_1+P_2-\nu+\bar{\nu},\\
(f_3)&=P_1-P_2+\nu-\bar{\nu},
\endaligned
$$
where $P_1$ and $P_2$ are points on $\mathcal{C}_2$ over $\lambda=\infty$, and $\nu\in\mathcal{C}_2$
is defined with $\nu_\lambda=
-\frac yx$, $\nu_\mu=-\omega\mid_{\lambda=-\frac yx}$.
\end{lem}

Now we are going to analyze the converse problem. Suppose the evolution
in time of the point $\nu$ is known. For reconstructing
the matrix $L(\lambda)$, one needs $x=|x|e^{i\arg x},\ y=|y|e^{i\arg y}$
as functions of time.

\begin{lem}\cite{DG}
The point $\nu\in\Gamma_2$ and the initial conditions for $M$ and $\Gamma$
determine $|x|,\ |y|$ and $\arg y-\arg x$, where $x$ and $y$ are given by
$\eqref{5.05}$.
\end{lem}

So, in order to determine $L(\lambda)$ as a function of time, one needs to find the evolution of the point $\nu$ and $\arg x$ as a function of time.
In \cite{DG} the following two theorems are proved:
\begin{thm}\cite{DG}
The integration of the motion of the point $\nu$ reduces to the inversion
of the elliptical integral
$$
\int_{\nu_0}^\nu\frac{d\lambda}{\sqrt{\Omega^2-2\Delta\Delta^*}}=
\frac 1{I_2}t.
$$
\end{thm}
Denote by $\phi_x=\arg
x$, and $u=tg \frac{\phi_x}{2}$.
\begin{thm}\cite{DG}
The function $u(t)$ satisfies the Riccati equation:
$$
\frac{du}{dt}=[f(t)+g(t)]u^2+[f(t)-g(t)],
$$
where
$$
\begin{aligned}
f(t)&=\frac K{2|x|^2},\qquad g(t)=\frac{Q|x|}{2},\\
K&=\frac{\langle M, \Gamma\rangle}{2\sqrt{X_0^2+Z_0^2}},\quad Q=\frac {\beta}{\alpha}\sqrt{2}\left(\frac 1I_2-\frac 1I_1\right);
\end{aligned}
$$
$|x|$ is a known function of time.
\end{thm}

As we presented in the previous subsection, the classical integration procedure
also yields one elliptic integral and the Riccati equation.

\subsection{Zhukovski's geometric interpretation}\label{x}

In \cite{Zh} Zhukovski gave a geometric interpretation of the Hess--Appel'rot conditions. Denote $J_i=1/I_i$.

Let us consider the so-called gyroscopic inertia ellipsoid:
$$
\frac {M_1^2}{J_1}+\frac {M_2^2}{J_2}+\frac
{M_3^2}{J_3}=1,
$$
and the plane containing the middle axis and  intersecting the
ellipsoid at a circle. Denote by $l$ the normal to the plane,
which passes through the fixed point $O$. Then the condition \eqref{ha}
means that the center of masses  lies on the line $l$.

If  we choose a basis of moving frame  such that the third axis is $l$, the second one is directed
along the middle axis of the ellipsoid, and the first one is
chosen according to the orientation of the orthogonal frame, then (see \cite{BM}), the invariant relation \eqref{4} becomes
$$
F_4=M_3=0,
$$
the matrix $J$ obtains the form:
$$
J=\left(\begin{matrix} J_1&0&J_{13}\\
                0&J_1&0\\
                J_{13}&0&J_3
        \end{matrix}\right),
$$
and ${\chi}=(0, 0, Z_0)$.

One can see here that the Hess--Appel'rot system can be regarded as a perturbation of the Lagrange top. In new coordinates the Hamiltonian of
the Hess--Appel'rot system becomes
$$
H_{HA}=\frac{1}{2}(J_1(M_1^2+M_2^2)+J_3M_3^2)+Z_0\Gamma_3+J_{13}M_1M_3=H_L+J_{13}M_1M_3
$$

This was used as a motivation for a definition of the higher-dimensional Hess--Appel'rot systems in \cite{DG2}.

\section{Higher-dimensional generalization}

Now we will pass to the higher-dimensional rigid body motion.
Let us consider motion of $N$ points in $\mathbb{R}^n$ such that the distance between each two of them is constant in time.
As an analogy with the three-dimensional case, we have two reference frames: the fixed and the moving ones.
In the moving reference frame, the velocity of the
$i$-th point is:
$$
V_i(t)=B^{-1}\dot{q}_i(t)=B^{-1}\dot{B}Q_i=\Omega(t)Q_i
$$
where again $Q_i$ represents the radius vector of the $i$-th point, and $\Omega$ is skew-symmetric matrix ($\Omega\in so(n)$) representing
the angular velocity of the body in the moving reference frame.
The angular momentum is a skew-symmetric matrix defined by
$$
\begin{aligned}
M&=\sum_i m_i (V_iQ_i^t-Q_iV_i^t)=\sum_i m_i(\Omega Q_i Q_i^t-Q_i Q_i^t\Omega^t)\\
&=\sum_i m_i(\Omega Q_i Q_i^t+Q_iQ_i^t\Omega)=\Omega I+I\Omega,
\end{aligned}
$$
where $I=\sum_i m_iQ_iQ_i^t$ is a constant symmetric matrix called the \emph{mass tensor of the body} (see \cite {FK}).

If one chooses the basis in which $I=\diag(I_1,...,I_n)$, the coordinates of angular momentum are $M_{ij}=(I_i+I_j)\Omega_{ij}$.

The kinetic energy is
$$
T=\frac12\sum_i m_i\langle \dot{Q}_i,\dot{Q}_i\rangle=\frac12\sum_i m_i\langle \Omega Q_i,\Omega Q_i\rangle.
$$
Since it is a homogeneous quadratic form of angular velocity $\Omega$, one has $\langle\frac{\partial T}{\partial\Omega}, \Omega\rangle=2T$
where $\langle A, B\rangle =-\frac12 Trace(AB)$ is an invariant scalar product on $so(n)$.
One gets
$$
\frac{\partial T}{\Omega_{kl}}=\sum_{m}(\Omega_{km}I_{ml}+I_{km}\Omega_{ml}),
$$
or $\frac{\partial T}{\partial Q}=M$
and finally
$$
T=\frac{1}{2}\langle M, \Omega\rangle.
$$

\begin{rem} The expression $M_{ij}=(I_i+I_j)\Omega_{ij}$ gives a left-invariant metric on $SO(n)$.
In that sense, the solutions of the classical Euler equations can be interpreted as geodesic lines of the left-invariant metric on $SO(3)$.
Arnol'd generalized the Euler equations (see \cite{Ar}).
He derived equations of geodesics of an arbitrary left-invariant metric on Lie group $G$.
In \cite{M} Manakov found L-A pair for the wider class of metrics $M_{ij}=\frac{a_i-a_j}{b_i-b_j}\Omega_{ij}$, and showed that this
class belongs to the class considered by Dubrovin in \cite{D1}, and hence, the solutions can be expressed in theta functions.
\end{rem}

The Lie group $E(3)$ can be regarded as a semidirect product of the Lie groups  $SO(3)$ and $\mathbb{R}^3$. The product in the group
given by
$$
(A_1,r_1)\cdot(A_2,r_2)=(A_1A_2, r_1+A_1r_2)
$$
corresponds to the composition of two isometric transformations of the Euclidian space. The Lie algebra $e(3)$ is a semidirect product
of $\mathbb{R}^3$ and $so(3)$. Using isomorphism between the Lie algebras $so(3)$ and $\mathbb{R}^3$, given by \eqref{so3}, one concludes
that $e(3)$ is also isomorphic to the semidirect product $s=so(3)\times_{ad} so(3)$. The commutator in $s$ is given by:
$$
[(a_1,b_1),(a_2,b_2)]=([a_1,a_2],[a_1,b_2]+[b_1,a_2]).
$$

One concludes, that there are two natural higher-dimensional generalizations of equations \eqref{EP}.
The first one is on the Lie algebra $e(n)$ that is a semidirect product of $so(n)$ and $\mathbb{R}^n$. The $n$-dimensional Lagrange case
on $e(n)$ is defined in \cite{B}, where its integrability is proved. The higher-dimensional Kowalevski case together with Lax representation
is constructed in \cite{BRS}
(see also \cite{BBIM}). For a list of integrable cases see for example \cite{TF}.

The second one, given by Ratiu in \cite{R} is
on semidirect product $s=so(n)\times_{ad}so(n)$.
Equations of motion in moving frame are (\cite{R}):
\begin{equation}
\dot M=\left[ M,\Omega \right] +\left[ \Gamma , \chi \right],\quad
\dot \Gamma =\left[ \Gamma , \Omega \right].
\label{6}
\end{equation}
Here $M\in so(n)$ is the angular momentum,
$\Omega \in so(n)$ is the angular velocity, $\chi \in so(n)$ is
a given constant matrix (describing  a generalized center of the
mass), $\Gamma \in so(n)$. Angular momentum $M$ and $\Omega$ are connected by $M=I\Omega+\Omega I$.
If the matrix $I$ is diagonal, $I=\diag(I_1,\dots ,I_n)$,
then $M_{ij}=(I_i+I_j)\Omega _{ij}$.
The Lie algebra $s$ is the Lie algebra of Lie group $S=SO(n)\times_{Ad}so(n)$ that is semidirect product of $SO(n)$ and $so(n)$ (here $so(n)$ is
considered  as the Abelian Lie group). The group product in $S$ is $(A_1,b_1)\cdot (A_2,b_2)=(A_1A_2, b_1+Ad_{A_1}b_2)$.

Ratiu proved that equations \eqref{6} are Hamiltonian in the Lie-Poisson structure on coadjoint orbits of group $S$ given by:
\begin{equation}
\begin{aligned}
\{\tilde f,\tilde g\}(\mu, \nu)=&-\mu([d_1f(\mu, \nu),d_1g(\mu, \nu)])\\
&-\nu([d_1f(\mu, \nu),d_2g(\mu, \nu)])\\
&-\nu([d_2f(\mu, \nu),d_1g(\mu, \nu)]),
\end{aligned}
\label{rps1}
\end{equation}
where  ${\tilde f}, {\tilde g}$ are restrictions of functions $f$ and $g$
on orbits of coadjoint action and $d_if$ are partial derivatives od $df$. On $so(n)$
a bilinear symmetric nondegenerate biinvariant (i.e. $k([\xi,\eta],\zeta)=k(\xi,[\eta,\zeta])$)
 two form exist, which can be extended to $s$ as well:
$$
k_s((\xi_1, \eta_1),(\xi_2, \eta_2))=k(\xi_1, \eta_2)+k(\xi_2,\eta_1).
$$
Hence, one can identify  $s^*$ and $s$.
Then, the Poisson structure \eqref{rps1} can be written in the form
\begin{equation}
\begin{aligned}
\{{\tilde f},{\tilde g}\}(\xi,\eta)=&-k(\xi, [(grad_2 f)(\xi,
\eta), (grad_1 g)(\xi,\eta)])\\
&-k(\xi, [(grad_1 f)(\xi,\eta), (grad_2 g)(\xi,\eta)])\\
&-k(\eta, [(grad_2 f)(\xi,\eta), (grad_2 g)(\xi,\eta)]),
\end{aligned}
\label{rps}
\end{equation}
where $grad_i$ are $k$-gradients in respect to the $i$-th coordinate.

In \cite{R}, the Lagrange case was defined by $I_1=I_2=a,\ I_3=\dots=I_n=b,\
\chi_{12}=-\chi_{21}\ne 0,\ \chi_{ij}=0,\ (i,j)\notin \{(1,2), (2,1)\}$. The completely
symmetric case was defined there by
$I_1=\dots=I_n=a$, where $\chi\in so(n)$ is an arbitrary constant matrix.
It was shown in \cite{R} that equations \eqref{6} in these cases could be
represented by the following L-A pair:
$$
\frac d{dt} (\lambda^2C+\lambda M+\Gamma)=[\lambda^2C+\lambda M+\Gamma, \lambda
\chi+\Omega],
$$
where in the Lagrange case $C=(a+b)\chi$, and in the symmetric case $C=2a\chi$.

\subsection{Four-dimensional rigid body motion}

To any $3\times 3$ skew-symmetric matrix one assigns one
vector in three-dimensional space using isomorphism between $\mathbb{R}^3$ and $so(3)$.
Using the the isomorphism between $so(4)$ and $so(3)\times so(3)$, one can assign two three-dimensional
vectors $A_1$ and $A_2$ to $(4\times4)$-skew-symmetric matrix $A$.

Vectors $A_1$ and $A_2$ are defined by:
$$
A_1=\frac{A_++A_-}{2},\quad A_2=\frac{A_+-A_-}{2},
$$
where
$A_+,A_-\in \mathbb{R}^3$ correspond to  $A_{ij}\in so(4)$ according to:
\begin{equation}
(A_+,A_-)\rightarrow
\left(\begin{matrix}
0 & -A^3_{+} & A^2_{+} & -A^1_{-}\\
A^3_{+} & 0 & -A^1_{+} & -A^2_{-}\\
-A^2_{+} & A^1_{+} & 0 & -A^3_{-}\\
A^1_{-} & A^2_{-} & A^3_{-} & 0
\end{matrix}\right)
\label{so4}
\end{equation}
Here $A_{\pm}^j$ are the $j$-th coordinates of the vector $A_\pm$.

By direct calculations, we check that vectors $2A_1\times B_1$ and $2A_2\times B_2$ correspond to commutator $[A,B]$, if
vectors $A_1, A_2$ and $B_1$, $B_2$ correspond to $A$ and $B$ respectively.

Consequently, equations of motion \eqref{6} on $so(4)\times so(4)$ can be written as:
\begin{equation}
\begin{aligned}
\dot M_1&=2(M_1\times\Omega_1+\Gamma_1\times\chi_1)\qquad
\dot \Gamma_1=2(\Gamma_1\times\Omega_1)\\
\dot M_2&=2(M_2\times\Omega_2+\Gamma_2\times\chi_2)\qquad
\dot \Gamma_2=2(\Gamma_2\times\Omega_2)
\end{aligned}
\label{4j}
\end{equation}
Recall that $M=I\Omega+\Omega I$. The matrix elements of the mass tensor of the body $I$
are $I_{kl}=\sum m_iQ_{(i)k}Q_{(i)l},\,k,l=1,...,4$.
Choose the coordinates $(X_1, X_2, X_3, X_4)$ of the moving reference frame in which $I$ has diagonal
form $I=\diag(I_1,I_2,I_3,I_4)$. Then, for example $I_1=\sum m_i X_{(i)1}^2$, $I_2=\sum m_i X_{(i)2}^2$, $\sum m_i X_{(i)1}X_{(i)2}=0$ etc.,
where $X_{(i)k}$ is the $k$-th coordinate of $i$-th point. In the three-dimensional case the moments of inertia were defined with respect to the
line through the fixed point $O$.
We derive the angular velocity $\omega$ as a skew-symmetric matrix the elements of which correspond to the rotations in
two-dimensional coordinate planes. Hence, here it is natural to define the moments of inertia of the body with respect to the two-dimensional
planes through the fixed point. For example the moment of inertia with respect to the plane $X_1OX_2$ is $I_1+I_2$, and
$M_{12}=(I_1+I_2)\Omega_{12}$, etc.

Here we observe a complete analogy with the three-dimensional case. For example, the moment of inertia with respect to $OZ$
axis $I_{33}=\sum_i m_i(X_i^2+Y_i^2)$
consists of two addend $\sum_i m_iX_i^2$ and  $\sum_i m_iY_i^2$ that are diagonal elements of the mass tensor of the body.

For vectors $M_+$ and $M_-$ one has
$$
\begin{aligned}
M_+&=\big((I_2+I_3)\Omega_{+}^1,
(I_3+I_2)\Omega_+^2, (I_3+I_1)\Omega_+^3\big)=I_+\Omega_+\\
M_-&=\big((I_1+I_4)\Omega_{-}^1,
(I_2+I_4)\Omega_-^2, (I_3+I_4)\Omega_-^3\big)=I_-\Omega_-.
\end{aligned}
$$
Finally, one can calculate
\begin{equation}
\begin{aligned}
M_1&=\frac12\big((I_++I_-)\Omega_1+(I_+-I_-)\Omega_2\big)\\
M_2&=\frac12\big((I_+-I_-)\Omega_1+(I_++I_-)\Omega_2\big)\\
\end{aligned}
\label{mio}
\end{equation}

At a glance it looks that \eqref{4j} are equations of motion of two independent three-dimensional rigid bodies.
However, the formulas \eqref{mio} show that they are not independent and that each of $M_1$, $M_2$ depends on both $\Omega_1$ and $\Omega_2$.

\subsection{Lagrange bitop. Definition and Lax representation}

Generalizing the Lax representation of the Hess--Appel'rot system, the new complete integrable
four-dimensional rigid body system is established in \cite{DG}. A detailed classical and algebro-geometric
integration was presented in \cite{DG1}.

The Lagrange bitop is four-dimensional rigid body system defined by (see \cite{DG, DG1}):
\begin{equation}
\begin{aligned}
I_1&=I_2=a\\
I_3&=I_4=b
\end{aligned}
 \quad
\rm{and}\quad \chi =
\left(\begin{matrix} 0 & \chi _{12} & 0 & 0 \\
-\chi _{12} & 0 & 0 & 0\\
0& 0 & 0 & \chi _{34}\\
0 & 0 & -\chi _{34} & 0
\end{matrix}
\right)
\label{7}
\end{equation}
with the conditions $a\neq b,\,\,\chi_{12}, \chi_{34}\neq0,
|\chi_{12}|\neq |\chi_{34}|$.

We have the following proposition:
\begin{prop} \cite{DG, DG1} The equations of motion \eqref{6} under conditions
\eqref{7} have an $L-A$ pair representation $\dot L(\lambda )
=\left[ L(\lambda ), A(\lambda )\right],$ where
\begin{equation}
L(\lambda )=\lambda ^2 C+\lambda M +\Gamma, \quad A(\lambda
)=\lambda \chi +\Omega,
\label{8}
\end{equation}
and $C=(a+b)\chi $.
\end{prop}

Let us briefly analyze spectral properties of the matrices $L(\lambda)$.
The spectral polynomial $p(\lambda, \mu )=\det \left (
L(\lambda )-\mu \cdot 1\right)$ has the form
\begin{equation*}
p(\lambda , \mu )=\mu ^4+P(\lambda )\mu ^2 +[Q(\lambda )]^2,
\end{equation*}
where
\begin{equation}
\begin{aligned}
P(\lambda )&=A\lambda ^4 +B\lambda ^3+D\lambda ^2+E\lambda +F,\\
Q(\lambda )&=G\lambda ^4+H\lambda ^3+I\lambda ^2+J\lambda +K.
\end{aligned}
\label{9.1}
\end{equation}
Their coefficients
\begin{align*}
A&=C_{12}^2+C_{34}^2=\langle C_+,C_+\rangle +\langle C_-,C_- \rangle,\\
B&=2C_{34}M_{34}+2C_{12}M_{12}=2\left( \langle C_+,M_+\rangle
+\langle C_-,
M_-\rangle \right),\\
D&=M_{13}^2+M_{14}^2+M_{23}^2+M_{12}^2+M_{34}^2+2C_{12}\Gamma
_{12}+2C_{34}
\Gamma _{34}\\
&=\langle M_+,M_+\rangle +\langle M_-,M_-\rangle +2\left( \langle
C_+,
\Gamma _+\rangle +\langle C_-,\Gamma _- \rangle \right), \\
E&=2\Gamma _{12}M_{12}+2\Gamma _{13}M_{13}+2\Gamma
_{14}M_{14}+2\Gamma _{23}
M_{23}+2\Gamma _{24}M_{24}+2\Gamma _{34}M_{34}\\
&=2\left( \langle \Gamma _+, M_+ \rangle +\langle \Gamma _- , M_-
\rangle
\right), \\
F&=\Gamma _{12}^2+\Gamma _{13}^2+\Gamma _{14}^2+\Gamma
_{23}^2+\Gamma _{24}^2 +\Gamma _{34}^2=\langle \Gamma _+ , \Gamma
_+ \rangle +\langle \Gamma _-,
\Gamma _-\rangle ,\\
G&=C_{12}C_{34}=\langle C_+, C_-\rangle,\\
H&=C_{34}M_{12}+C_{12}M_{34}=\langle C_+,M_-\rangle +\langle C_-, M_+\rangle,\\
I&=C_{34}\Gamma _{12}+\Gamma _{34}C_{12}+M_{12}M_{34}+M_{23}M_{14}-M_{13}M_{24}\\
&=\langle C_+, \Gamma _- \rangle +\langle C_-,\Gamma _+\rangle
+\langle M_+,
M_- \rangle , \\
J&=M_{34}\Gamma _{12}+M_{12}\Gamma _{34}+M_{14}\Gamma
_{23}+M_{23}\Gamma _{14}-
\Gamma _{13}M_{24}-\Gamma _{24}M_{13}\\
&=\langle M_+, \Gamma _- \rangle +\langle M_-, \Gamma _+ \rangle, \\
K&=\Gamma _{34}\Gamma _{12}+\Gamma _{23}\Gamma _{14}-\Gamma
_{13}\Gamma _{24} =\langle \Gamma _+, \Gamma _-\rangle.
\end{align*}
are integrals of motion of the system \eqref{6}, \eqref{7}. Here
 $M_+,M_-\in \mathbb{R}^3$ are defined with \eqref{so4} (similar for other vectors).
  System \eqref{6}, \eqref{7} is Hamiltonian with the Hamiltonian
function
$$
{\mathcal{H}}=\frac12(M_{13}\Omega_{13}+M_{14}\Omega_{14}+M_{23}\Omega_{23}+
M_{12}\Omega_{12}+M_{34}\Omega_{34})+\chi_{12}\Gamma
_{12}+\chi_{34}\Gamma_{34}.
$$
The algebra $so(4)\times so(4)$ is 12-dimensional. The general
orbits of the coadjoint action are 8-dimensional. According to
\cite{R}, the Casimir functions are coefficients of $\lambda^0,
\lambda, \lambda^4$ in the polynomials $[\det \tilde L(\lambda
)]^{1/2}$ and $-\frac 12 Tr (\tilde L(\lambda ))^2$.
One calculates:
$$
\left[ det \tilde L(\lambda )\right] ^{1/2} =G\lambda ^4+H\lambda ^3 +I\lambda ^2 +
J\lambda +K,\quad
-\frac 12 Tr\left(\tilde L(\lambda )\right)^2 =A\lambda ^4+E\lambda +F.
$$
So, Casimir functions are $J,K,E,F$. Nontrivial integrals of motion are $B,D,H,I$, and, one can check that they
are in involution. When $|\chi_{12}|=|\chi_{34}|$, then $2H=B$ or
$2H=-B$ and there are only 3 independent integrals in involution.
Thus,
\begin{prop}\cite{DG1}
For $|\chi_{12}|\ne |\chi_{34}|$, system \eqref{6},
\eqref{7} is completely integrable in the Liouville sense.
\end{prop}
System \eqref{6}, \eqref{7} doesn't fall in any of
the families defined by Ratiu in \cite{R} and together with them it makes complete list of
systems with the $L$ operator of the form
$$
L(\lambda)=\lambda ^2C+\lambda M + \Gamma .
$$
More precisely, if $\chi_{12}\ne 0$, then the Euler-Poisson
equations \eqref{6} could be written in the form \eqref{8} (with
arbitrary $C$) if and only if equations \eqref{6} describe the
generalized symmetric case, the generalized Lagrange case or the
Lagrange bitop, including the case $\chi_{12}=\pm \chi_{34}$
\cite{DG}.

\subsubsection{Classical integration}

For classical integration we will use equations \eqref{4j}. On can calculate that
$$
\chi_1=(0,0,-\frac12 (\chi_{12}+\chi_{34})),\quad
\chi_2=(0,0,-\frac12 (\chi_{12}-\chi_{34}))
$$
and also
$$
\begin{aligned}
M_1&=((a+b)\Omega_{(1)1}, (a+b)\Omega_{(1)2}, (a+b)\Omega_{(1)3}+
(a-b)\Omega_{(2)3})\\
M_2&=((a+b)\Omega_{(2)1}, (a+b)\Omega_{(2)2}, (a-b)\Omega_{(1)3}+
(a+b)\Omega_{(2)3}).
\end{aligned}
$$

If we denote $\Omega_1=(p_1, q_1, r_1),\ \Omega_2=(p_2, q_2, r_2)$,
then the first group of the equations \eqref{4j} becomes
$$
\begin{aligned}
&{\dot p}_1-mq_1r_2=-n_1\Gamma_{(1)2}, &&{\dot p}_2-mq_2r_1=-n_2\Gamma_{(2)2}\\
&{\dot q}_1+mp_1r_2=n_1\Gamma_{(1)1},  &&{\dot q}_2+mp_2r_1=n_2\Gamma_{(2)1}\\
&(a+b){\dot r}_1+(a-b){\dot r}_2=0, &&(a-b){\dot r}_1+(a+b){\dot r}_2=0
\end{aligned}
$$
where
$$
m=-\frac{2(a-b)}{a+b},\qquad n_1=-\frac{2\chi_{(1)3}}{a+b},\qquad
n_2=-\frac{2\chi_{(2)3}}{a+b}.
$$
The integrals of motion are for $i=1,2$:
$$
\begin{aligned}
&(a+b)\alpha_i\chi_{(i)3}=f_{i1}\\
&(a+b)[(a+b)(p_i^2+q_i^2)+(a+b)\alpha_i^2+2\chi_{(i)3}\Gamma_{(i)3}]
=f_{i2}\\
&(a+b)p_i\Gamma_{(i)1}+(a+b)q_i\Gamma_{(i)2}+(a+b)\alpha_i\Gamma_{(i)3}
=f_{i3}\\
&\Gamma_{(i)1}^2+\Gamma_{(i)2}^2+\Gamma_{(i)3}^2=1,
\end{aligned}
$$
where
$$
\begin{aligned}
\alpha_1&=\frac{(a+b)r_1+(a-b)r_2}{a+b}\quad
\alpha_2=\frac{(a+b)r_2+(a-b)r_1}{a+b}\\
a_i&=\frac{\alpha_i^2(a+b)^2-f_{i2}}{(a+b)^2}\quad
i=1,2.
\end{aligned}
$$
Introducing $\rho_i, \sigma_i$, defined with $p_i=\rho_i\cos\sigma_i$,
$q_i=\rho_i\sin\sigma_i$, after calculations, one gets
\begin{equation}
\begin{aligned}
&\rho_1^2{\dot\sigma}_1+mr_2\rho_1^2=n_1(\frac{f_{13}}{a+b}-\alpha_1
\Gamma_{(1)3})\\
&[(\rho_i^2)^{\cdot}]^2=4n_i^2\rho_i^2[1-\frac{1}{n_i^2}(a_i+\rho_i^2)^2]-
4n_i^2(\frac{f_{i3}}{a+b}-\alpha_ia_i-\frac{\alpha_i}{n_i}\rho_i^2)^2, \quad i=1,2\\
&\rho_2^2{\dot\sigma}_2+mr_1\rho_2^2=n_2(\frac{f_{23}}{a+b}-\alpha_2
\Gamma_{(2)3}).
\end{aligned}
\label{16}
\end{equation}
Let us denote $u_1=\rho_1^2,\ u_2=\rho_2^2$. From \eqref{16} we have
$$
{\dot u}_i^2= P_i(u_i),\qquad i=1,2,
$$
$$
P_i(u)=-4u^3-4u^2B_i+4uC_i+D_i,\qquad i=1,2;
$$
$$
\begin{aligned}
B_i&=2a_i+\alpha_i^2,\quad
C_i=n_i^2-a_i^2- 4\frac{\alpha_i\chi_{(i)3}f_{i3}}{(a+b)^2}-2\alpha_i^2a_i,\\
D_i&=-4(\frac{2\chi_{(i)3}f_{i3}}{(a+b)^2}+\alpha_ia_i)^2,\qquad i=1,2.
\end{aligned}
$$
From the previous relations, we have
$$
\int\frac{du_1}{\sqrt{P_1(u_1)}}=t,\quad
\int\frac{du_2}{\sqrt{P_2(u_2)}}=t.
$$
So,  the integration of the Lagrange bitop leads to the functions associated
with the elliptic curves $E_1, E_2$
where $E_i=E_i(\alpha_i, a_i, \chi_{(i)3}, f_{i2}, f_{i3})$ are
given with:
\begin{equation}
E_i:  y^2 = P_i(u).
\label{2ec}
\end{equation}
Equations \eqref{4j} are very similar to those for the classical
Lagrange system. However, the system doesn't split
on two independent Lagrangian systems

\subsubsection{Properties of spectral curve}

The spectral curve is given by:
$$
\mathcal{C}:\ \ \mu ^4+P(\lambda )\mu ^2 +[Q(\lambda )]^2=0
$$
where $P$ and $Q$ are given by \eqref{9.1}.

There is an involution $\sigma:\;(\lambda,\mu)\rightarrow
(\lambda, -\mu)$ on the spectral curve  which corresponds to the
skew symmetry of the matrix $L(\lambda)$. Denote the factor-curve
by  $\mathcal{C}_1=\mathcal{C}/\sigma$.

\begin{lem}\cite{DG1}
\begin{itemize}
\item The curve $\mathcal{C}_1$ is a smooth hyperelliptic curve of the
genus $g(\mathcal{C}_1)=3$. The spectral curve $\mathcal{C}$ is a double covering of
$\mathcal{C}_1$. The arithmetic genus of $\mathcal{C}$ is
$g_a(\mathcal{C})=9$.
\item The spectral curve $\mathcal{C}$ has four
ordinary double points $S_i, i=1,\dots, 4$. The genus of its
normalization $\tilde{\mathcal{C}}$ is five.
\item  The singular
points $S_i$ of the curve $\mathcal{C}$ are fixed points of the
involution $\sigma$. The involution $\sigma $ exchanges the two
branches of $\mathcal{C}$ at $S_i$.
\end{itemize}
\end{lem}

In general, whenever matrix $L(\lambda)$ is skew-symmetric, the
spectral curve is reducible in an odd-dimensional case and singular
in an even-dimensional case.

The detailed algebro-geometric integration procedure of the system is given in \cite{DG1}.
Analysis of the spectral curve and the Baker--Akhiezer function
shows that the dynamics of the system is related to a certain Prym
variety $\Pi$ that corresponds to the double covering defined by the involution $\sigma$ and
to evolution of divisors of some meromorphic
differentials $\Omega^i_j$. It appears that
$$
\Omega^1_2,\;\Omega^2_1,\;\Omega^3_4,\;\Omega^4_3
$$
are {\it holomorphic} during the whole evolution. Compatibility of this
requirement with the dynamics puts a strong constraint on the spectral curve:
{\it its theta divisor should contain some torus}. In the case presented here
such a constraint appears to be satisfied according to Mumford's relation.
These conditions create a new
situation from the point of view of the existing integration
techniques. For details see \cite{DG1}.

\subsection{Four-dimensional Hess--Appel'rot systems}

The starting point for construction of generalization of the Hess--Appel'rot system was Zhukovski's geometric interpretation given in subsection \ref{x}.
Having it in mind, in \cite{DG2} the higher-dimensional Hess--Appel'rot systems are defined.
First we will consider the four-dimensional case on $so(4)\times so(4)$. We will consider metric given with $\Omega=JM+MJ$.

\begin{dfn}\label{d1}\cite{DG2} The four-dimensional Hess--Appel'rot system
is described by equations \eqref{6} and satisfies the conditions:
\begin{enumerate}
\item
\begin{equation}
\Omega=MJ+JM,\quad J=\left (\begin{matrix} J_1&0&J_{13}&0\\
                    0&J_1&0&J_{24}\\
                    J_{13}&0&J_3&0\\
                    0&J_{24}&0&J_3
\end{matrix}\right)
\end{equation}
\item
$$
\chi=\left(\begin{matrix} 0&\chi_{12}&0&0\\
                    -\chi_{12}&0&0&0\\
                    0&0&0&\chi_{34}\\
                    0&0&-\chi_{34}&0
                    \end{matrix}\right).
$$
\end{enumerate}
\end{dfn}

The invariant surfaces are determined in the following lemma.

\begin{lem}\cite{DG2} For the four-dimensional Hess--Appel'rot system,
the following relations take place:
$$
\begin{aligned}
\dot M_{12}&=J_{13}(M_{13}M_{12}+M_{24}M_{34})+J_{24}(M_{13}M_{34}+M_{12}M_{24}),\\
\dot M_{34}&=J_{13}(-M_{13}M_{34}-M_{12}M_{24})+J_{24}(-M_{13}M_{12}-M_{24}M_{34}).
\end{aligned}
$$
In particular, if $M_{12}=M_{34}=0$ hold at the initial moment,
then the same relations are  satisfied during the  evolution in
time.
\end{lem}

Thus, in the four-dimensional Hess--Appel'rot case, there are two
invariant relations
\begin{equation}
M_{12}=0,\quad M_{34}=0.
\label{ir4}
\end{equation}

Let us now present another definition of the four-dimensional
Hess--Appel'rot conditions, starting from a basis where the matrix
$J$ is diagonal in.

Let $\tilde J=\diag(\tilde J_1,
\tilde J_2, \tilde J_3, \tilde J_4)$.

\begin{dfn}\label{d2}\cite{DG2} The  four-dimensional Hess--Appel'rot
system is
 described by the equations \eqref{6} and satisfies the conditions:
\begin{enumerate}
\item
$$
\Omega=M\tilde J+\tilde J M,\ \ \tilde J=\diag(\tilde J_1, \tilde J_2,
\tilde J_3, \tilde J_4),
$$
\item
$$
\tilde\chi=\left(\begin{matrix} 0&\tilde\chi_{12}&0&\tilde\chi_{14}\\
                    -\tilde\chi_{12}&0&\tilde\chi_{23}&0\\
                    0&-\tilde\chi_{23}&0&\tilde\chi_{34}\\
                    -\tilde\chi_{14}&0&-\tilde\chi_{34}&0
                    \end{matrix}\right),
$$
\item
$$
\begin{aligned}
\tilde J_3-\tilde J_4 &=\tilde J_2-\tilde J_1,\\
\frac{\tilde{J_3}-\tilde{J_1}}{\sqrt{1+t_1^2}}&=\frac{\tilde{J_4}-\tilde{J_2}}{\sqrt{1+t_2^2}}\\
\end{aligned}
$$
\end{enumerate}
where
$$
\begin{aligned}
t_1 & :=\frac {2(\tilde \chi_{14}\tilde \chi_{34}-\tilde \chi_{12}\tilde
\chi_{23})}{\tilde \chi_{14}^2-\tilde \chi_{34}^2+\tilde \chi_{12}^2-\tilde
\chi_{23}^2},\\
t_2 & :=\frac {2(\tilde \chi_{14}\tilde \chi_{12}-\tilde
\chi_{23}\tilde \chi_{34})}{-\tilde \chi_{14}^2-\tilde
\chi_{34}^2+\tilde \chi_{12}^2+\tilde
\chi_{23}^2}.\\
\end{aligned}
$$
\end{dfn}

\begin{prop}\cite{DG2} There exists a bi-correspondence between sets of data from Definition \ref{d1} and Definition \ref{d2}.
\end{prop}

\begin{rem} 1) In the case $J_{24}\ne 0, \chi_{34}=0$, there is an additional relation $\tilde \chi_{12}\tilde \chi_{34}+\tilde \chi_{14}\tilde
\chi_{23}=0$.
It follows from the system
$$
\begin{aligned}
\tilde\chi_{12}\sin\varphi+\tilde\chi_{23}\cos\varphi &=0,\\
\tilde\chi_{14}\sin\varphi-\tilde\chi_{34}\cos\varphi &=0,\\
\end{aligned}
$$

2) In the  case $J_{24}= 0, \chi_{34}=0$, additional relations are $\tilde \chi_{34}=\tilde \chi_{14}=0$,
and the second relation from Definition \ref{d2} can be replaced by
the relation
$$
\tilde\chi_{12}\sqrt{\tilde J_2-\tilde J_1}+
\tilde\chi_{23}\sqrt{\tilde J_3-\tilde J_2}=0.
$$
\end{rem}

\begin{thm} \cite{DG2} The four-dimensional Hess--Appel'rot system
 has the following Lax representation
$$
\begin{aligned}
\dot L(\lambda)&=[L(\lambda), A(\lambda)],\\
L(\lambda)=\lambda^2 C+ \lambda M+\Gamma,\ \ A(\lambda)&=\lambda\chi +\Omega,
\ \ C=\frac1{J_1+J_3}\chi.
\end{aligned}
$$
\end{thm}

One can calculate the spectral polynomial for the four-dimensional
Hess--Appel'rot system:
$$
p(\lambda, \mu)=\det(L(\lambda)-\mu\cdot 1)=
\mu^4+P(\lambda)\mu^2+Q(\lambda)^2,
$$
where
$$
\begin{aligned}
P(\lambda)&=a\lambda^4+b\lambda^3+c\lambda^2+d\lambda+e\\
Q(\lambda)&=f\lambda^4+g\lambda^3+h\lambda^2+i\lambda+j
\end{aligned}
$$
$$
\begin{aligned}
a&=C_{12}^2+C_{34}^2,\\
b&=2C_{12}M_{12}+2C_{34}M_{34}(=0),\\
c&=M_{13}^2+M_{14}^2+M_{23}^2+M_{24}^2+M_{12}^2+M_{34}^2+2C_{12}\Gamma _{12}+2C_{34}\Gamma_{34},\\
d&=2\Gamma _{12}M_{12}+2\Gamma _{13}M_{13}+2\Gamma _{14}M_{14}+2\Gamma _{23}
M_{23}+2\Gamma _{24}M_{24}+2\Gamma _{34}M_{34}\\
e&=\Gamma _{12}^2+\Gamma _{13}^2+\Gamma _{14}^2+\Gamma _{23}^2+\Gamma _{24}^2
+\Gamma _{34}^2,\\
f&=C_{12}C_{34}\\
g&=C_{12}M_{34}+C_{34}M_{12}(=0),\\
h&=\Gamma _{34}C_{12}+\Gamma_{12}C_{34}+M_{12}M_{34}+M_{23}M_{14}-M_{13}M_{24},\\
i&=M_{34}\Gamma _{12}+M_{12}\Gamma _{34}+M_{14}\Gamma _{23}+M_{23}\Gamma _{14}-
\Gamma _{13}M_{24}-\Gamma _{24}M_{13},\\
j&=\Gamma _{34}\Gamma _{12}+\Gamma _{23}\Gamma _{14}-\Gamma _{13}\Gamma _{24}.
\end{aligned}
$$
In the standard Poisson structure on semidirect product $so(4)\times so(4)$
the functions $d,e,i,j$ are Casimir functions, $c, h$
are first integrals, and $b=0, g=0$ are the invariant relations.
As we already mentioned general orbits of co-adjoint action are eight-dimensional, thus for
complete integrability one needs four independent integrals in involution.

\subsection {The $n$-dimensional Hess--Appel'rot systems}

In \cite{DG2} we introduced also Hess--Appel'rot systems of arbitrary
dimension.

\begin{dfn}\label{d3} The $n$-dimensional Hess--Appel'rot system
is described by the equations \eqref{6}, and satisfies the
conditions:
\begin{enumerate}
\item
$$
\Omega=JM+MJ,\ \ J=\left (\begin{matrix} J_1&0&J_{13}&0&0&...&0\\
                                0&J_1&0&J_{24}&0&...&0\\
                                J_{13}&0&J_3&0&0&...&0\\
                                0&J_{24}&0&J_3&0&...&0\\
                                0&0&0&0&0&...&0\\
                                .&.&.&.&.&...&.\\
                                .&.&.&.&.&...&.\\
                                0&0&0&0&0&...&J_3
                                \end{matrix}\right),
$$
\item
$$
\chi=\left (\begin{matrix} 0&\chi_{12}&0&...&0\\
                    -\chi_{12}&0&0&...&0\\
                    0&0&0&...&0\\
                    0&0&0&...&0\\
                    .&.&.&...&.\\
                    .&.&.&...&.\\
                    0&0&0&...&0
                    \end{matrix}\right).
$$
\end{enumerate}
\end{dfn}
Direct calculations give the following lemma:

\begin{lem}\cite{DG2} For the $n$-dimensional Hess--Appel'rot system,
the following relations are satisfied:
\begin{enumerate}
\item
$$
\begin{aligned}
\dot M_{12}&=J_{13}(M_{12}M_{13}+M_{24}M_{34}+\sum_{p=5}^n M_{2p}M_{3p})+\\
&J_{24}(M_{12}M_{24}+M_{13}M_{34}-\sum_{p=5}^nM_{1p}M_{4p})\\
\dot M_{34}&=-J_{13}(M_{13}M_{34}+M_{24}M_{12}+\sum_{p=5}^nM_{1p}
M_{p4})-\\
&J_{24}(M_{13}M_{12}+M_{24}M_{34}+\sum_{p=5}^nM_{2p} M_{3p}),\\
\dot M_{3p}&=-J_{13}(M_{13}M_{3p}+M_{2p}M_{12})-J_{24}(M_{34}M_{2p}+M_{23}M_{4p})+\\
&M_{34}\Omega_{4p}-\Omega_{34}M_{4p}+\sum_{k=5}^n(M_{3k}\Omega_{kp}-\Omega_{3k}M_{4p}),\
p>4,\\
\dot M_{4p}&=J_{13}(-M_{14}M_{3p}+M_{1p}M_{34})+J_{24}(M_{12}M_{1p}-M_{24}M_{4p})-\\
&M_{34}\Omega_{3p}+\Omega_{34}M_{3p}+\sum_{k=5}^n(M_{4k}\Omega_{kp}-\Omega_{4k}M_{4p}),\
p>4,\\
\end{aligned}
$$
\item
$$
\dot M_{kl}=0,\ \ \ k,l>4.
$$
\item
The $n$-dimensional Hess--Appel'rot case has the following
system of invariant relations
$$
M_{12}=0,\ \ M_{lp}=0,\ \ l,p\ge 3.
$$
\end{enumerate}
\end{lem}
By diagonalizing the matrix $J$, we come to another  definition

\begin{dfn}\label{d4}\cite{DG2} The $n$-dimensional Hess--Appel'rot system
is described by the equations \eqref{6}, and satisfies the
conditions
\begin{enumerate}
\item
$$
\Omega=\tilde J M+M\tilde J,\ \ \tilde J=\diag(\tilde J_1, \tilde
J_2, \tilde J_3, \tilde J_4,...,\tilde J_4),
$$
\item
$$
\tilde\chi=\left (\begin{matrix} 0&\tilde\chi_{12}&0&\tilde\chi_{14}&...&0\\
                    -\tilde\chi_{12}&0&\tilde \chi_{23}&0&...&0\\
                    0&-\tilde \chi_{23}&0&\tilde\chi_{34}&...&0\\
                    -\tilde\chi_{14}&0&-\tilde\chi_{34}&0&...&0\\
                    .&.&.&.&...&.\\
                    .&.&.&.&...&.\\
                    0&0&0&0&...&0
                    \end{matrix}\right),
$$
\item
$$
\begin{aligned}
\tilde J_3-\tilde J_4 &=\tilde J_2-\tilde J_1,\\
\frac{\tilde{J_3}-\tilde{J_1}}{\sqrt{1+t_1^2}}&=
\frac{\tilde{J_4}-\tilde{J_2}}{\sqrt{1+t_2^2}}
\\
\tilde\chi_{12}\tilde\chi_{34}&+\tilde\chi_{14}\tilde\chi_{23}=0
\end{aligned}
$$
where
$$
\begin{aligned} t_1 & :=\frac {2(\tilde \chi_{14}\tilde \chi_{34}-\tilde
\chi_{12}\tilde \chi_{23})}{\tilde \chi_{14}^2-\tilde
\chi_{34}^2+\tilde \chi_{12}^2-\tilde
\chi_{23}^2},\\
t_2 & :=\frac {2(\tilde \chi_{14}\tilde \chi_{12}-\tilde
\chi_{23}\tilde \chi_{34})}{-\tilde \chi_{14}^2-\tilde
\chi_{34}^2+\tilde \chi_{12}^2+\tilde
\chi_{23}^2}.\\
\end{aligned}
$$
\end{enumerate}
\end{dfn}

As in the dimension four, there is an equivalence of the definitions.

\begin{prop}\cite{DG2} There exists a bi-correspondence between
sets of data from Definition \ref{d3} and Definition \ref{d4}.
\end{prop}

The following theorem gives a Lax pair for the $n$-dimensional Hess--Appel'rot system.

\begin{thm} \cite{DG2} The $n$-dimensional Hess--Appel'rot system
has the following Lax pair
$$
\begin{aligned}
\dot L(\lambda)&=[L(\lambda), A(\lambda)],\\
L(\lambda)=\lambda^2 C+ \lambda M+\Gamma,\ \ A(\lambda)&=\lambda\chi +\Omega,
\ \ C=\frac1{J_1+J_3}\chi.
\end{aligned}
$$
\end{thm}

\subsection{Classical integration of the four-dimensional Hess--Appel'rot system.}

Detailed classical and algebro-geometric integration procedures for the four-di\-men\-si\-o\-nal Hess-Appel'rot case
are presented in \cite{DG2}. Here again equations \eqref{4j} are useful for classical integration.
We have:
$$
\chi_1=(0,0,-\frac12 (\chi_{12}+\chi_{34})),\quad
\chi_2=(0,0,-\frac12 (\chi_{12}-\chi_{34})).
$$
Integrals of the motion are
\begin{equation}
\begin{aligned}
\langle M_i, M_i\rangle+2\frac{1}{J_1+J_3}\langle \chi_i, \Gamma_i\rangle&=h_i,\\
\langle\Gamma_i,\Gamma_i\rangle&=1,\qquad i=1,2,\\
\langle M_i,\Gamma_i\rangle&=c_i,\\
\langle\chi_i,M_i\rangle&=0.
\end{aligned}
\label{i4}
\end{equation}
Here the metric that gives connections between $M$ and $\Omega$ is different from that for Lagrange bitop. We have
$$
\begin{aligned}
\Omega_1&=( (J_1+J_3)M_{(1)1}-(J_{13}-J_{24})M_{(2)3}, (J_1+J_3)M_{(1)2},\\
&(J_1+J_3)M_{(1)3}+(J_1-J_3)M_{(2)3}-(J_{13}+J_{24})M_{(2)1}),\\
\Omega_2&=( (J_1+J_3)M_{(2)1}-(J_{13}+J_{24})M_{(1)3}, (J_1+J_3)M_{(2)2},\\
&(J_1+J_3)M_{(2)3}+(J_1-J_3)M_{(1)3}-(J_{13}-J_{24})M_{(1)1}),
\end{aligned}
$$
where again $M_{(i)j}$ is the $j$-th component of the vector $M_i$. Using
these expressions, equations \eqref{4j} can be rewritten in the
following form:
\begin{equation}
\begin{aligned}
\dot M_{(1)1}=&2[ (J_1-J_3)M_{(1)2}M_{(2)3}-(J_{13}+J_{24})M_{(1)2}M_{(2)1}
+\Gamma_{(1)2}\chi_{(1)3}],\\
\dot M_{(1)2}=&2[-(J_1-J_3)M_{(2)3}M_{(1)1}-(J_{13}-J_{24})M_{(1)3}M_{(2)3}+\\
&(J_{13}+J_{24})M_{(1)1}M_{(2)1}-\Gamma_{(1)1}\chi_{(1)3}],\\
\dot M_{(1)3}=&2(J_{13}-J_{24})M_{(1)2}M_{(2)3},\\
\dot \Gamma_{(1)1}=&2[\Gamma_{(1)2}((J_1+J_3)M_{(1)3}+(J_1-J_3)M_{(2)3}-
(J_{13}+J_{24})M_{(2)1})-\\
&\Gamma_{(1)3}(J_1+J_3)M_{(1)2}],\\
\dot \Gamma_{(1)2}=&2[\Gamma_{(1)3}((J_1+J_3)M_{(1)1}-(J_{13}-J_{24})M_{(2)3})-\\
&\Gamma_{(1)1}((J_1+J_3)M_{(1)3}+(J_1-J_3)M_{(2)3}-(J_{13}+J_{24})M_{(2)1})],\\
\dot \Gamma_{(1)3}=&2[\Gamma_{(1)1}(J_1+J_3)M_{(1)2}-\Gamma_{(1)2}((J_1+J_3)M_{(1)1}-
(J_{13}-J_{24})M_{(2)3})],
\end{aligned}
\label{4haa}
\end{equation}
and
\begin{equation}
\begin{aligned}
\dot M_{(2)1}=&2[ (J_1-J_3)M_{(2)2}M_{(1)3}-(J_{13}-J_{24})M_{(2)2}M_{(1)1}
+\Gamma_{(2)2}\chi_{(2)3}],\\
\dot M_{(2)2}=&2[-(J_1-J_3)M_{(1)3}M_{(2)1}-(J_{13}+J_{24})M_{(2)3}M_{(1)3}+\\
&(J_{13}-J_{24})M_{(2)1}M_{(1)1}-\Gamma_{(2)1}\chi_{(2)3}],\\
\dot M_{(2)3}=&2(J_{13}+J_{24})M_{(2)2}M_{(1)3},\\
\dot \Gamma_{(2)1}=&2[\Gamma_{(2)2}((J_1+J_3)M_{(2)3}+(J_1-J_3)M_{(1)3}-
(J_{13}-J_{24})M_{(1)1})-\\
&\Gamma_{(2)3}(J_1+J_3)M_{(2)2}],\\
\dot \Gamma_{(2)2}=&2[\Gamma_{(2)3}((J_1+J_3)M_{(2)1}-(J_{13}+J_{24})M_{(1)3})-\\
&\Gamma_{(2)1}((J_1+J_3)M_{(2)3}+(J_1-J_3)M_{(1)3}-(J_{13}-J_{24})M_{(1)1})],\\
\dot \Gamma_{(2)3}=&2[\Gamma_{(2)1}(J_1+J_3)M_{(2)2}-\Gamma_{(2)2}((J_1+J_3)M_{(2)1}-
(J_{13}+J_{24})M_{(1)3})].
\end{aligned}
\label{4hab}
\end{equation}
One can see here  that
$M_{(1)3}=M_{(2)3}=0$, giving two invariant relations introduced
before.

Let us introduce coordinates $K_i$ and $l_i$ as follows:
$$
M_{(i)1}=K_i\sin l_i,\qquad M_{(i)2}=K_i\cos l_i,\qquad i=1,2.
$$
From equations \eqref{4haa}, \eqref{4hab}, using integrals \eqref{i4}, we have
$$
\dot\Gamma_{(1)3}^2=4(J_1+J_3)^2\left[(1-\Gamma_{(1)3}^2)
(h_1-\frac2{J_1+J_3}\chi_{(1)3}\Gamma_{(1)3})-c_1^2\right]=
P_3(\Gamma_{(1)3}).
$$
Thus $\Gamma_{(1)3}$ can be solved by an elliptic quadrature. Also
from the energy integral we have that
$$
K_1^2=h_1-\frac2{J_1+J_3}\chi_{(1)3}\Gamma_{(1)3}.
$$
Since $\tan l_1=\frac{M_{(1)1}}{M_{(1)2}}$, we have:
$$
\dot l_1=-2(J_{13}+J_{24})K_2\sin l_2+\frac{2\chi_{(1)3}c_1}{K_1^2}.
$$
and
$$
K_1^2\Gamma_{(1)2}^2-2c_1M_{(1)2}\Gamma_{(1)2}+c_1^2-M_{(1)1}^2(1-\Gamma_{(1)3}^2)=0.
$$
Similarly, one gets:
$$
\begin{aligned}
\dot\Gamma_{(2)3}^2&=4(J_1+J_3)^2\left[(1-\Gamma_{(2)3}^2)
(h_2-\frac2{J_1+J_3}\chi_{(2)3}\Gamma_{(2)3})-c_2^2\right]=
P_3(\Gamma_{(2)3}),\\
K_2^2&=h_2-\frac2{J_1+J_3}\chi_{(2)3}\Gamma_{(2)3},\\
\dot l_2&=-2(J_{13}-J_{24})K_1\sin l_1+\frac{2\chi_{(2)3}c_2}{K_2^2},\\
K_2^2\Gamma_{(2)2}^2&-2c_2M_{(2)2}\Gamma_{(2)2}+c_2^2-M_{(2)1}^2(1-\Gamma_{(2)3}^2)=0.
\end{aligned}
$$

From the previous considerations one concludes that
integration of the four-dimensional Hess--Appel'rot system leads to a system of two differential equations (for $l_1$ and
$l_2$) of the first order and two elliptic integrals,
associated with elliptic curves $E_1$ and $E_2$ defined by
$$
E_i:\quad y^2=P_i(x)=8A_ix^3-4B_ix^2-8A_ix-4C_i,\ \ i=1,2
$$
where
$$
A_i=(J_1+J_3)\chi_{(i)3}, \ B_i=(J_1+J_3)^2h_i,\
C_i=(J_1+J_3)^2(c_i^2-h_i).
$$
This is a typical situation for the Hess--Appel'rot systems that
additional integrations are required.

In \cite{DG2} the algebro-geometric integration procedure is presented. It is closely related
to the integration of the Lagrange bitop.

In \cite{J, J1} the Hess--Appel'rot systems are considered within the framework of partial reduction.

\section{Another Lax representation for the classical Hess--Appel'rot case. Generalizations}

In 1846 Jacobi gave an algebraic description of the Jacobian of a hyperelliptic curve. Beauville noticed that using given
description, any hyperelliptic curve can be seen as a
spectral curve of some matrix $L(\lambda)$. Starting from the spectral curve for the Lagrange top, given by Ratiu and van Moerbeke,
Gavrilov and Zhivkov in \cite{GZ} have constructed a new L-A pair for the Lagrange top. Using a modification
of that construction, starting from elliptic curve \eqref{5.1} we have constructed
another Lax representation for the Hess--Appel'rot system. This was a starting point for
construction of a class of systems on $e(3)$, with the same elliptic curve as a spectral curve (for details see \cite{DG5}).
Let $\omega, \Delta$ and $\Delta^*$ are given by \eqref{5.05}.

\begin{prop} On hypersurface \eqref{4} the equations of the Hess--Appel'rot system
are equivalent to the Lax representation
\begin{equation}
\dot{L}(\lambda)=\frac{1}{2I_2}\left[L(\lambda),\frac{\lambda^2
L(\lambda)-a^2L(a)}{\lambda-a}\right].
\label{ha1}
\end{equation}
where
\begin{equation}
L(\lambda)=\left[\begin{matrix}\frac{\omega(\lambda)}{\lambda^2}&\sqrt{2}i\frac{\Delta(\lambda)}{\lambda^2}\\
 & \\
\sqrt{2}i\frac{\Delta^*(\lambda)}{\lambda^2}&-\frac{\omega(\lambda)}{\lambda^2}\end{matrix}\right],
\label{ha2}
\end{equation}
and
 $a=\frac{\alpha\Omega_1+\beta\Omega_2}{\sqrt{X_0^2+Z_0^2}}$.
\end{prop}

Starting from Lax representation \eqref{ha1} in \cite{DG5} the family of new systems
is described. Let us consider the general case of equations \eqref{ha1}, with $a$ as an arbitrary
polynomial in $M$'s and $\Gamma$'s.

The corresponding spectral curve is the elliptic curve
\begin{equation}
\mu^2=\frac{\omega^2(\lambda)}{\lambda^4}-2\frac{\Delta(\lambda)\Delta^*(\lambda)}{\lambda^4}.
\label{ha3}
\end{equation}

Here $M$ ad $\Gamma$ are in $e(3)$ as before. In terms of $x, y, \bar x, \bar y, x_1, y_1$
from \eqref{5.05} and $y_1=\alpha\Gamma_1+\beta\Gamma_3,\ x_1=\alpha M_1+\beta M_3$,
the standard Poisson structure \eqref{PS} on $e(3)$ has the form:
$$
\begin{aligned}
\{x,y\}&=0,\ \   \{\bar{x}, \bar{y}\}=0, \ \  \{x,x_1\}=ix, \ \  \{\bar{x},x_1\}=-i\bar{x}, \ \ \{y,x_1\}=iy,\\
\{\bar{y},x_1\}&=-i\bar{y},  \ \ \{y_1, x_1\}=0,  \ \ \{\bar{x},y_1\}=-i\bar{y}, \ \ \{x,y_1\}=iy,  \ \ \{y_1, y\}=0, \\
 \{y_1, \bar{y}\}&=0,\ \   \{x, \bar{x}\}=-ix_1,  \ \ \{y,\bar{y}\}=0,  \ \ \{x,\bar{y}\}=-iy_1,  \ \ \{\bar{x}, y\}=iy_1.
\end{aligned}
$$

Observe that matrices $L$ given by \eqref{ha2} satisfy
\begin{equation*}
\left\{\overset {1}{L}(\lambda), \overset{2}{L}(\mu)\right\}=
\left[r(\lambda-\mu),
\overset{1}{L}(\lambda)+\overset{2}{L}(\mu)\right],
\end{equation*}
where
$$
\overset{1}{L}(\lambda)=L(\lambda)\otimes\left[\begin{matrix} 1&0\\ 0&1\end{matrix}\right],\ \ \
\overset{2}{L}(\mu)=\left[\begin{matrix} 1&0\\ 0&1\end{matrix}\right]\otimes L(\mu),
$$
with the permutation matrix as an $r$-matrix
$$
r(\lambda)=\frac{-1}{\lambda}
\left[\begin{matrix} 1&0&0&0\\
                     0&0&1&0\\
                     0&1&0&0\\
                     0&0&0&1
\end{matrix}\right].
$$

Equations \eqref{ha1} can be rewritten in the form
\begin{equation}
\begin{aligned}
\dot{M}_1&=Z_0\Gamma_2 +aZ_0M_2,\\
\dot{M}_2&=X_0\Gamma_3-Z_0\Gamma_1+a(X_0M_3-Z_0M_1),\\
\dot{M}_3&=-X_0\Gamma_2 -aX_0M_2,\\
\dot{\Gamma}_1&=\frac{\Gamma_2M_3-\Gamma_3M_2}{I_2}+aZ_0\Gamma_2,\\
\dot{\Gamma}_2&=\frac{\Gamma_3M_1-\Gamma_1M_3}{I_2}+a(X_0\Gamma_3-Z_0\Gamma_1),\\
\dot{\Gamma}_3&=\frac{\Gamma_1M_2-\Gamma_2M_1}{I_2}-aX_0\Gamma_2.
\end{aligned}
\label{ha7}
\end{equation}

We have the following Proposition.
\begin{prop} System \eqref{ha7} can be rewritten as:
\begin{equation*}
\begin{aligned}
\dot{M}_i&=\{M_i, H_1\}+a\{M_i, H_2\},\\
\dot{\Gamma}_i&=\{\Gamma_i, H_1\}+a\{\Gamma_i, H_2\},\quad i=1,2,3,
\end{aligned}
\end{equation*}
where
$$
H_1=\frac{M_1^2+M_2^2+M_3^2}{2I_2}+(X_0\Gamma_1+Z_0\Gamma_3),\quad
H_2=X_0M_1+Z_0M_3.
$$
\end{prop}

As we have already mentioned the Poisson bracket \eqref{PS} has two Casimir functions:
\begin{equation*}
F_1=M_1\Gamma_1+M_2\Gamma_2+M_3\Gamma_3,\ \ F_2=\Gamma_1^2+\Gamma_2^2+\Gamma_3^2.
\end{equation*}
Thus, a symplectic leaf, defined by conditions  $F_1=c_1, F_2=c_2$
is a four-dimensional manifold. For integrability in the Liouville
sense on $e(3)$, another first integral beside the Hamiltonian is
necessary. On the other hand, if a system is not Hamiltonian,
generally speaking, five first integrals of motion for
integrability in quadratures are required. But, if a
nonhamiltonian system has an  invariant measure, then, according
to the Jacobi theorem, for integrability in quadratures one needs four first
integrals of motion.

For a general polynomial $a$, the system \eqref{ha7} is neither
Hamiltonian in the Poisson structure \eqref{PS}, nor preserves the
standard measure. A simple criterion for preserving the standard
measure is given by:
\begin{prop}\cite{DG5} System \eqref{ha7} preserves the standard measure if and only if the polynomial
$a$ satisfies the condition:
$$
\{a, X_0M_1+Z_0M_3\}=0.
$$
\end{prop}

As a consequence we have:

\begin{prop}\label{p6}\cite{DG5}
In the following five cases, the  standard measure is preserved
\begin{description}
\item[(i)] if the polynomial $a$ is a Casimir function:
$a=M_1\Gamma_1+M_2\Gamma_2+M_3\Gamma_3$;
\item[(ii)] if the
polynomial $a$ is a Casimir function:
$a=\Gamma_1^2+\Gamma_2^2+\Gamma_3^2$;
\item[(iii)] if
$a=X_0M_1+Z_0M_3$; \item[(iv)] if $a=X_0\Gamma_1+Z_0\Gamma_3$;
\item[(v)]if $a=M_1^2+M_2^2+M_3^2$.
\end{description}
\end{prop}
\begin{thm}\cite{DG5} If $X_0\ne0$, or $Z_0\ne 0$, in the first three cases given above,
the systems are Hamiltonian, while in the fourth and the fifth
cases, the systems are not Hamiltonian in the Poisson structure
\eqref{PS}.
\end{thm}
If $a$ is a Casimir function, for an arbitrary function $f$ we
have
$$
\{f,H_1\}+a\{f,H_2\}=\{f,H_1+aH_2\}.
$$
Hence, in the first two cases the systems are Hamiltonian with
Hamiltonian functions
$$
H=H_1+aH_2.
$$
In the third case, since $a=H_2$, we have
$$
\{x^i,H_1\}+H_2\{x^i,H_2\}=\{x^i, H_1+\frac{H_2^2}{2}\},
$$
where $x^i, \ i=1,...,6$ are coordinates $M_1,M_2,M_3,\Gamma_1,
\Gamma_2, \Gamma_3$. Thus, the system is also Hamiltonian with the
Hamiltonian function
$$
H=H_1+\frac{H_2^2}{2}.
$$

Regarding integrability of the given five cases, we have the simple Proposition.
\begin{prop}\cite{DG5}
\item{\rm{(a)}} A function $F$ is a first integral of equations \eqref{ha7} if it satisfies
$$
\dot{F}=\{F,H_1\}+a\{F,H_2\}=0.
$$
\item{\rm{(b)}} The Casimir functions $F_1$ and $F_2$ and
functions $H_1$ and $H_2$ are integrals of system \eqref{ha7} for
any polynomial $a$.
\end{prop}
Finally we have:
\begin{thm}\cite{DG5} System \eqref{ha7} in cases (i)-(iii) is completely integrable in the Liouville sense.
In cases (iv) and (v), system \eqref{ha7} is integrable in quadratures.
\end{thm}

\subsection{Algebro-geometric integration procedure of the systems}

The  algebro-geometric integration procedure for the first three cases (i)-(iii) (when system \eqref{ha7} is Hamiltonian)
is done in \cite{DG5}. It is based on
a construction of the Baker--Akhiezer vector-function.

As  usual, we consider the following eigenvalue problem
\begin{equation*}
\begin{aligned}
&\left(\frac{d}{dt}+A(\lambda)\right)\Psi(t,P)=0,\\
&L(\lambda)\Psi(t,P)=\mu\Psi(t,P),
\end{aligned}
\end{equation*}
with a normalization
\begin{equation*}
\Psi^1(0,P)+\Psi^2(0,P)=1,
\end{equation*}
where $P=(\lambda, \mu)$ is a point on the spectral curve
$\mathcal{C}$.

Let us denote by $\infty^+$ and $\infty^-$ the two points on the
curve $\mathcal{C}$ over $\lambda=\infty$, with
$\mu=iI_2\sqrt{X_0^2+Z_0^2}$ and $\mu=-iI_2\sqrt{X_0^2+Z_0^2}$
respectively.
\begin{prop}\label{p8} \cite{DG5} If the polynomial $a(M,\Gamma)$ is  a first integral of  motion, then
the vector-function $\Psi(t,P)$ satisfies the following
conditions:
\begin{itemize}
\item[(a)] In the affine part of the curve $\Gamma$, the
vector-function $\Psi(t,P)$ has two time independent poles, and
each of the components $\Psi^1(t,P)$ and $\Psi^2(t,P)$ has one
zero.
\item[(b)] At the points $\infty^+$ and $\infty^-$, the
functions $\Psi^1$ and $\Psi^2$ have essential singularities with
the following asymptotics:

\begin{equation*}
\Psi^1(t,P)=\begin{cases}e^{\frac{i}{2}(\sqrt{X_0^2+Z_0^2}\ (\lambda+a)+\frac{x_1}{I_2})t}
(1+O(\frac{1}{\lambda})),  &\textrm{for }P\rightarrow \infty^-\\
e^{-\frac{i}{2}(\sqrt{X_0^2+Z_0^2}\ (\lambda+a)+\frac{x_1}{I_2})t}
(O(\frac{1}{\lambda})),  &\textrm{for }P\rightarrow \infty^+
\end{cases}
\end{equation*}
\begin{equation*}
\Psi^2(t,P)=\begin{cases}e^{\frac i2(\sqrt{X_0^2+Z_0^2}\ (\lambda+a)+\frac{x_1}{I_2})t}
\left(O(\frac 1\lambda)\right), &\textrm{for } P\rightarrow \infty^-\\
e^{-\frac i2\left(\sqrt{X_0^2+Z_0^2}\ (\lambda+a)+\frac{x_1}{I_2}\right)t}
\left(1+O(\frac 1\lambda)\right),  &\textrm{for }P\rightarrow \infty^+\\
\end{cases}
\end{equation*}
\item[(c)] The asymptotics have the form
$$
\begin{aligned}
\Psi^1(t,P)&=e^{-\frac i2\left(\sqrt{X_0^2+Z_0^2}\ (\lambda+a)+\frac{x_1}{I_2}\right)t}
\left(\frac{x}{I_2\sqrt{2}\sqrt{X_0^2+X_0^2}}\frac 1\lambda+O(1/\lambda^2)\right), \,\, P\rightarrow \infty^+\\
\Psi^2(t,P)&=e^{\frac i2\left(\sqrt{X_0^2+Z_0^2}\ (\lambda+a)+\frac{x_1}{I_2}\right)t}
\left(-\frac{\bar{x}}{I_2\sqrt{2}\sqrt{X_0^2+X_0^2}}\frac 1\lambda+O(1/\lambda^2)\right), \,\, P\rightarrow \infty^-\\
\end{aligned}
$$
\end{itemize}
\end{prop}

Now we will give explicit formulae for the Baker--Akhiezer function
in terms of the Jacobi theta-function $\theta_{11}(z|\tau)$ with
characteristics $[\frac12,\frac12]$.

Let us fix the canonical basis of cycles $A$ and $B$ on $\Gamma$
($A\cdot B=1$), and let $\omega$ be the holomorphic differential
normalized by the conditions
$$
\oint_A\omega=2i\pi,\ \ \oint_B\omega=\tau.
$$
A theta-function $\theta_{11}(z|\tau)$ is defined by the relation
$$
\theta_{11}(z|\tau)=\sum_{-\infty}^{\infty} \exp\left[ \frac 12\tau(n+\frac12)+(z+i\pi)(n+\frac12)\right].
$$
Let $\Omega^+$ and $\Omega^-$ be differentials of the second kind
with principal parts $-\frac{i}{2}\sqrt{X_0^2+Z_0^2}\ d\lambda$ and
$+\frac{i}{2}\sqrt{X_0^2+Z_0^2}\ d\lambda$ at $\infty^+$ and at
$\infty^-$ respectively, normalized by the condition that
$A$-periods are zero. Let us introduce differential $\Omega=\Omega^++\Omega^-$. We will denote by $U$ the
$B$-period of differential $\Omega$, and by
$c^+$ and $c^-$ the constants:
$$
\begin{aligned}
\int_{P_0}^P\Omega&=-\frac{i}{2}\sqrt{X_0^2+Z_0^2}\ \lambda+c^++O(1/\lambda),\ \ P\rightarrow P^+\\
\int_{P_0}^P\Omega&=+\frac{i}{2}\sqrt{X_0^2+Z_0^2}\ \lambda+c^-+O(1/\lambda),\ \ P\rightarrow P^-.
\end{aligned}
$$
\begin{prop} \cite{DG5} The Baker--Akhiezer functions are given by
$$
\begin{aligned}
\Psi^1(t,P)&\\
=&c_1\exp\left[(\int_{P_0}^P\Omega-c^-+\frac{i}{2}a+\frac{i}{2}\frac{x_1}{I_2})t\right]
\frac{\theta_{11}(\mathcal{A}(P+\infty^+-P_1-P_2)+tU)}{\theta_{11}({\mathcal{A}(\infty^++\infty^--P_1-P_2)+tU})},\\
\Psi^2(t,P)&\\
=&c_2\exp\left[(\int_{P_0}^P\Omega-c^+-\frac{i}{2}a-\frac{i}{2}\frac{x_1}{I_2})t\right]
\frac{\theta_{11}(\mathcal{A}(P+\infty^--P_1-P_2)+tU)}{\theta_{11}({\mathcal{A}(\infty^++\infty^--P_1-P_2)+tU})},\\
\end{aligned}
$$
where constants $c_1$ and $c_2$ are
$$
\begin{aligned}
c_1&=\frac{\theta_{11}(\mathcal{A}(P-\infty^+))\theta_{11}(\mathcal{A}(\infty^--P_1))
\theta_{11}(\mathcal{A}(\infty^--P_2))}{\theta_{11}(\mathcal{A}(\infty^--\infty^+))
\theta_{11}(\mathcal{A}(P-P_1))\theta_{11}(\mathcal{A}(P-P_2))},\\
c_2&=\frac{\theta_{11}(\mathcal{A}(P-\infty^-))\theta_{11}(\mathcal{A}(\infty^+-P_1))
\theta_{11}(\mathcal{A}(\infty^+-P_2))}{\theta_{11}(\mathcal{A}(\infty^+-\infty^-))
\theta_{11}(\mathcal{A}(P-P_1))\theta_{11}(\mathcal{A}(P-P_2))},\\
\end{aligned}
$$
and $\mathcal{A}$ is the Abel map, and $P_1$ and $P_2$ are the
poles of the  function $\Psi$.
\end{prop}

\subsection{Classical integration of the systems}

In new coordinates
$$
\begin{aligned}
X_1&=\alpha M_1+\beta M_3 & X_2&=M_2 & X_3&=-\beta M_1+\alpha M_3\\
Y_1&=\alpha \Gamma_1+\beta \Gamma_3 & Y_2&=\Gamma_2  & Y_3&=-\beta \Gamma_1+\alpha \Gamma_3
\end{aligned}
$$
differential equations of motion \eqref{ha7} become
\begin{equation*}
\begin{aligned}
\dot{X}_1&=0\\
\dot{X}_2&=\sqrt{X_0^2+Z_0^2}\,(Y_3+aX_3)\\
\dot{X}_3&=-\sqrt{X_0^2+Z_0^2}\,(Y_2+aX_2)\\
\dot{Y}_1&=\frac{1}{I_2}(X_3Y_2-X_2Y_3)\\
\dot{Y}_2&=\frac{1}{I_2}(X_1Y_3-X_3Y_1)+a\sqrt{X_0^2+Z_0^2}\,Y_3\\
\dot{Y}_3&=\frac{1}{I_2}(X_2Y_1-X_1Y_2)-a\sqrt{X_0^2+Z_0^2}\,Y_2.
\end{aligned}
\end{equation*}

The first integrals are
\begin{equation}
\begin{aligned}
F_1&=X_1Y_1+X_2Y_2+X_3Y_3=c_1\\
F_2&=Y_1^2+Y_2^2+Y_3^2=c_2\\
H_1&=\frac{X_1^2+X_2^2+X_3^2}{2I_2\sqrt{X_0^2+Z_0^2}}+Y_1=d_1\\
H_2&=X_1=d_2.
\end{aligned}
\label{ha12}
\end{equation}
Introducing  polar coordinates $X_2=\rho\cos{\sigma},\quad X_3=\rho\sin{\sigma}$
and using integrals \eqref{ha12}, after simplifying, and denoting $\rho^2=u$, one has
\begin{equation}
\dot{u}^2=-\frac{u^3}{I_2^2}-Bu^2-Cu-D
\label{ha13}
\end{equation}
where
$$
\begin{aligned}
B&=\frac{-4A\sqrt{X_0^2+Z_0^2}}{I_2}+\frac{d_2^2}{I_2^2}\\
C&=4(X_0^2+Z_0^2)\Big(A^2-c_2+\frac{d_2(c_1-d_2A)}{I_2\sqrt{X_0^2+Z_0^2}}\Big)\\
D&=4(X_0^2+Z_0^2)(c_1-d_2A)^2\\
A&=d_1-\frac{d_2^2}{2I_2\sqrt{X_0^2+Z_0^2}}.
\end{aligned}
$$
So, the following proposition is proved:
\begin{prop} The function $u(t)$ is an elliptic function of time.
\end{prop}
Let us remark that $u$ (and consequently $\rho$) does not depend
on a choice of the polynomial $a$.

Having $u(t)$ as a known function of time, one can find $\rho(t)$ as a known function of time. In order to
reconstruct $X_2$ and $X_3$, one needs to find $\sigma$ as a function of time.

We have
$$
\dot{\sigma}=-\frac{1}{\rho^2(t)}\sqrt{X_0^2+Z_0^2}\left[c_1-d_2\left(d_1-\frac{d_2^2+\rho^2(t)}{2I_2\sqrt{X_0^2+Z_0^2}}\right)+a\,\rho^2(t)\right]
$$
The right-hand side of the last equation is a function of time and
of the polynomial $a$. When $a$ is a first integral of  motion, then
the right hand side of the last equation is a known function of time.
Hence, one can find $\sigma$ by quadratures. In the fourth case
$$
a=X_0\Gamma_1+Z_0\Gamma_3=\sqrt{X_0^2+Z_0^2}\,Y_1=\sqrt{X_0^2+Z_0^2}\,d_1-\frac{d_2^2+\rho^2(t)}{2I_2\sqrt{X_0^2+Z_0^2}}.
$$
So, in this case $a(t)$ is a known function of time and one can
find $\sigma$ by solving a differential equation. Similarly, in
the fifth case
$$
a=M_1^2+M_2^2+M_3^2=X_1^2+X_2^2+X_3^2=d_2^2+\rho^2(t)
$$
is again a known function of time and a differential equation for
determining  $\sigma$ can be solved. Knowing $\rho$ and $\sigma$
as functions of time, one can easily reconstruct $X_2, X_3, Y_1, Y_2$ and $Y_3$.

Two elliptic curves appeared here. The first one $\mathcal{C}$, has
been defined by the equation \eqref{ha3}, and it was the curve
from which we started. The other one $\mathcal{C}'$, given by
\begin{equation}
v^2=-\frac{u^3}{I_2^2}-Bu^2-Cu-D
\label{ha14}
\end{equation}
corresponds to the solution of differential equation
\eqref{ha13}. A natural question is how these two curves are
related. We have the following proposition:

\begin{prop} \cite{DG5} The elliptic curves
$\mathcal{C}$, defined by equation
\eqref{ha3} and $\mathcal{C}'$ defined by \eqref{ha14} are
isomorphic.
\end{prop}

Using the Sklyanin magic recipe, from Lax representation \eqref{ha1} the
separation variables for the cases(i)-(iii) are constructed in \cite{DG5}. Also a sort of
separation variables for the Hess--Appel'rot case are found.

\section{Motion of rigid body in ideal fluid. Kirchhoff equations}

The mechanical system similar to the motion of a heavy rigid body fixed at a point is a
motion of a rigid body in an ideal incompressible fluid that is at rest at infinity.
The equations of motion are derived by Kirchhoff in 1870 (see \cite{K}).
They can be written in the form
\begin{equation}
\begin{aligned}
 \dot {\vec{M}}&={\vec M}\times {\frac{\partial H}{\partial \vec M}}+
{\vec\Gamma}\times{\frac{\partial H}{\partial \vec \Gamma}},\\
\dot {\vec \Gamma}&={\vec \Gamma}\times{\frac{\partial H}{\partial \vec M}}
\end{aligned}
\label{ke}
\end{equation}
where Hamiltonian $H$ is homogeneous quadratic function in $\vec M$ and $\vec\Gamma$ given by:
$$
H=\frac{1}{2}\langle A\vec M,\vec M\rangle+\langle B\vec M, \vec\Gamma\rangle+\frac{1}{2}\langle C\vec\Gamma,\vec\Gamma\rangle.
$$
Here $\vec{M}$ is \emph{impulsive moment} and $\vec{\Gamma}$ is \emph{impulsive force}. The matrix $A$ is positive-definite, the matrices $B$ and $C$
are symmetric.
Equations \eqref{ke} are Hamiltonian in the standard Lie-Poisson
structure given by \eqref{PS}. Hence for complete integrability in the Liouville sense one needs one additional independent integral.

The equations of motion of a heavy rigid body fixed at a point \eqref{EP} can be written in form \eqref{ke}
with $H$ from \eqref{PI} as a Hamiltonian function.

\subsection{Integrable cases}

We will list the integrable cases. For a full list and details see for example \cite{BM}.

The first nontrivial integrable case of equations \eqref{ke} was discovered by Kirchhoff in 1870 (see \cite{K}).
It is defined with conditions:

$\bullet$ Kirchhoff's case (1870):
$$
A=\diag(a_1, a_1, a_3),\quad B=\diag(b_1, b_1, b_3),\quad C=\diag(c_1, c_1, c_3).
$$
An additional integral is $F_4=M_3$. It is analogous to the Lagrange case of motion of a heavy rigid
body fixed at a point.

$\bullet$ The first Clebsch case (1871):
$$
A=\diag (a,a,a),\quad B=0,
$$
The additional integral is:
$$
F_4=a\langle C\vec M, \vec M\rangle-\det(C)\langle C^{-1}\vec\Gamma, \vec\Gamma\rangle.
$$

$\bullet$ The second Clebsch case (1871):
\begin{equation}
\begin{aligned}
A=\diag(&a_1, a_2, a_3),\quad B=0,\quad C=\diag(c_1,c_2,c_3)\\
&\frac{c_2-c_3}{a_1}+\frac{c_3-c_1}{a_2}+\frac{c_1-c_2}{a_3}=0.
\end{aligned}
\label{cleb}
\end{equation}

Conditions \eqref{cleb} are equivalent to:
$$
\frac{c_2-c_3}{a_1(a_2-a_3)}=\frac{c_3-c_1}{a_2(a_3-a_1)}=\frac{c_1-c_2}{a_3(a_1-a_2)}=\theta,
$$
where $a_1, a_2, a_3$ are pairwise distinct. The additional integral is:
$$
F_4=\theta\langle \vec M,\vec M\rangle-\langle A\vec\Gamma,\vec\Gamma\rangle.
$$

$\bullet$ Steklov's case (1893):
$$
\begin{aligned}
A&=\diag(a_1, a_2, a_3),\quad B=\diag(\mu a_2a_3,\mu a_3a_1, \mu a_1a_2),\\
C&=\diag(\mu^2a_1(a_2-a_3)^2,\mu^2a_2(a_3-a_1)^2,\mu^2a_3(a_1-a_2)^2)
\end{aligned}
$$
where $\mu$ is a constant. The additional integral is:
$$
F_4=\sum\limits_j (M_j^2-2\mu a_jM_j\Gamma_j)+\mu^2\big((a_2-a_3)^2\Gamma_1^2+(a_3-a_1)^2\Gamma_2^2+(a_1-a_2)^2\Gamma_3^2\big).
$$

$\bullet$ Lyapunov's case (1893):
$$
\begin{aligned}
A&=\diag(1,1,1),\quad B=\diag(-2\mu d_1,-2\mu d_2, -2\mu d_3),\\
C&=\diag(\mu^2(d_2-d_3)^2,\mu^2(d_3-d_1)^2,\mu^2(d_1-d_2)^2)
\end{aligned}
$$
The additional integral is:
$$
\begin{aligned}
F_4=&\sum\limits_j d_jM_j^2+2\mu(d_2d_3M_1\Gamma_1+d_3d_1M_2\Gamma_2+d_1d_2M_3\Gamma_3)+\\
&\mu^2(d_1(d_2-d_3)^2\Gamma_1^2+d_2(d_3-d_1)^2\Gamma_2^2+d_3(d_1-d_2)^2\Gamma_3^2)
\end{aligned}
$$

$\bullet$ Sokolov's case (2001):

$$
\begin{aligned}
a_1&=a_2=1,\quad a_3=2,\quad b_{13}=\alpha,\quad b_{23}=\beta,\\
c_{12}&=-4\alpha\beta,\quad c_{11}=4\beta^2,\quad
c_{22}=4\alpha^2,\quad c_{33}=-4(\alpha^2+\beta^2),
\end{aligned}
$$
and additional integral is
$$
F_4=(M_3-\alpha \Gamma_1-\beta \Gamma_2)^2P+Q^2,
$$
where
$$
\begin{aligned}
P=&(\alpha^2+\beta^2)(M_3+2\alpha \Gamma_1+2\beta \Gamma_2)^2+(\beta M_1-\alpha M_2)^2\\
Q=&\big[\alpha M_1+\beta M_2+(\alpha^2+\beta^2)\Gamma_3\big](M_3+2\alpha \Gamma_1+2\beta \Gamma_2)+\\
&3(\beta M_1-\alpha M_2)(\beta \Gamma_1-\alpha \Gamma_2).
\end{aligned}
$$

$\bullet$ Chaplygin's first case (1902):
$$
A=\diag(a, a, 2a), \quad B=0,\quad C=\diag(c, -c, 0),
$$
On the symplectic leaf given with $\langle \vec M, \vec\Gamma\rangle=0$, the equations
admit additional integral:
$$
F_4=(M_1^2-M_2^2+c\Gamma_3^2)^2+4M_1^2M_2^2.
$$

$\bullet$ Chaplygin's second case (1897).

Chaplygin's second case had an invariant relation instead of a fourth integral.
It is defined in 1897 by Chaplygin (see \cite{Ch}). This system was also
considered by Kozlov and Onischenko in \cite{KO}. It is defined by:
\begin{equation}
\begin{aligned}
A&=\diag(a_1, a_2, a_3)\\
b_{13}&\sqrt{a_2-a_1}\mp(b_2-b_1)\sqrt{a_3-a_2}=0,\, b_{12}=0\\
b_{13}&\sqrt{a_3-a_2}\pm(b_3-b_2)\sqrt{a_2-a_1}=0,\, b_{23}=0\\
c_{13}&\sqrt{a_2-a_1}\mp(c_2-c_1)\sqrt{a_3-a_2}=0,\, c_{12}=0\\
c_{13}&\sqrt{a_3-a_2}\pm(c_3-c_2)\sqrt{a_2-a_1}=0,\, c_{23}=0.
\end{aligned}
\label{cc}
\end{equation}

The invariant relation is: $F_4=M_1\sqrt{a_2-a_1}\mp M_3\sqrt{a_3-a_2}=0$.

Conditions \eqref{cc} may be regarded as  analogy of the
Hess--Appel'rot conditions in the case of motion of a heavy rigid
body fixed at a point. We have shown that Hess--Appel'rot case can be considered as a perturbation of the Lagrange top.

Similarly, the Chaplygin case is a perturbation of the Kirchhoff
case. If one chooses the basis where $a_1=a_2$,
the Chaplygin conditions become (see for example \cite{DG3, BM}):
$$
a_1=a_2,\ a_{13}\ne 0,\ B=\diag (b_1, b_1, b_3),\ C=\diag (c_1, c_1, c_3).
$$
In  new coordinates the invariant relation is $M_3=0$.

In the case $B=0$, Kirchhoff's case can be regarded as a special case of the Clebsch case. In \cite{Pe} Perelomov constructed the Lax representation
for the Clebsch case as well as higher-dimensional generalizations. Using this Lax representation in \cite{DG3} the Lax representation is
constructed for the  Chaplygin's second case:

\begin{thm}\cite{DG3} When $B=0$, on the invariant manifold given by the invariant relation,
the equations of motion of the  Chaplygin's second case
are equivalent to the matrix equation:
$$
\dot{L}(\lambda)=[L(\lambda), Q(\lambda)]
$$
where $ L(\lambda)=\lambda^2L_2+\lambda L_1-L_0$,
$Q(\lambda)=\lambda Q_1+Q_0$, and
$$
L_2=diag(c_1/a_1, c_1/a_1, c_3/a_1),\  Q_1=diag(a_1, a_1, a_3)
$$
$$
L_1=\left[\begin{matrix}
0&-M_3&M_2\\
M_3&0&-M_1\\
-M_2&M_1&0
\end{matrix}
\right]
\quad
L_0=
\Gamma\Gamma^{T}
$$
$$
Q_0=\left[\begin{matrix}
0&-a_3M_3-a_{13}M_1&a_1M_2\\
a_3M_3+a_{13}M_1&0&-a_1M_1-a_{13}M_3\\
-a_1M_2&a_1M_1+a_{13}M_3&0
\end{matrix}
\right]
$$
\end{thm}

The spectral curve $det(L(\lambda)-\mu\cdot 1)=0$ is
$$
\begin{aligned}
\mathcal{C}:\qquad \mu^3&+\mu^2F_3-\lambda_1^2\mu^2(c_3+2c_1)+\\
&\lambda_1^2\mu[2F_1-(2c_1+c_3)F_3]+\lambda_1^4\mu c_1(c_1+2c_3)-\\
&\lambda_1^6 c_1^2
c_3-\lambda_1^4(2c_1F_1-c_1(c_1+c_3)F_3)+\lambda_1^2a_1F_2^2=0,
\end{aligned}
$$
where $\lambda_1=\frac{\lambda}{\sqrt{a_1}}$.
It is singular and has an involution $\sigma:(\lambda_1,\mu)\to
(-\lambda_1,\mu)$. The curve $\mathcal{C}_1=\mathcal{C}/\sigma$ is a
nonsingular genus one curve.

\subsection{Four-dimensional Kirchhoff and Chaplygin cases}

In \cite{DG3} the four-di\-men\-si\-o\-nal generalization of the Kirchhoff and Chaplygin cases is constructed on $e(4$).

Let us consider the Hamiltonian equations with Hamiltonian
function:
$$
2H=\sum A_{ijkl}M_{ij}M_{kl}+2\sum B_{ijk}M_{ij}\Gamma_k+\sum C_{kl}\Gamma_k\Gamma_l
$$
in the standard Lie-Poisson structure on $e(4)$ given by:
$$
\{M_{ij}, M_{kl}\}=\delta_{ik}M_{jl}+\delta_{jl}M_{ik}-\delta_{il}M_{jk}-\delta_{jk}M_{il}
$$
$$
\{M_{ij},\Gamma_k\}=\delta_{ik}\Gamma_j-\delta_{jk}\Gamma_i
$$
A four-dimensional
Kirchhoff case should have two linear first integrals: $M_{12}$ and
$M_{34}$. It is interesting that under such assumption, the "mixed"
term in the Hamiltonian  is missing.
\begin{prop}\cite{DG3}
If $M_{12}$ and $M_{34}$ are the first integrals, then $B_{ijk}=0$.
\end{prop}
The proof follows through direct calculations.
\begin{dfn}
The four-dimensional Kirchhoff case is defined by
$$
\begin{aligned}
2H_K=&A_{1212}M_{12}^2+A_{1313}(M_{13}^2+M_{14}^2+M_{23}^2+M_{24}^2)+A_{3434}M_{34}^2+\\
&A_{1234}M_{12}M_{34}+C_{11}(\Gamma_1^2+\Gamma_2^2)+C_{33}(\Gamma_3^2+\Gamma_4^2)
\end{aligned}
$$
\end{dfn}
On $e(4)$ the standard Lie - Poisson structure has  two Casimir
functions:
\begin{equation*}
\begin{aligned}
F_1=&\Gamma_1^2+\Gamma_2^2+\Gamma_3^2+\Gamma_4^2,\\
F_2=&(M_{13}\Gamma_4-M_{14}\Gamma_3+M_{34}\Gamma_1)^2+(M_{23}\Gamma_1+M_{12}\Gamma_{3}-M_{13}\Gamma_2)^2+\\
&(M_{24}\Gamma_1-M_{14}\Gamma_2+M_{12}\Gamma_4)^2+(M_{23}\Gamma_4+M_{34}\Gamma_2-M_{24}\Gamma_3)^2
\end{aligned}
\end{equation*}
consequently, the general  symplectic leaves are 8-dimensional. For complete
integrability one needs four first integrals in involution. In \cite{DG3} it is proved that except
Hamiltonian, the four-dimensional Kirchhoff case has two linear first integrals $F_3=M_{12}$,
$F_4=M_{34}$ and one additional quadratic first integral:
\begin{equation*}
\begin{aligned}
F_5&=a_1(M_{12}M_{34}+M_{14}M_{23}-M_{13}M_{24})^2\\
&-c_1((M_{13}\Gamma_4-M_{14}\Gamma_3+M_{34}\Gamma_1)^2+(M_{23}\Gamma_4+M_{34}\Gamma_2-M_{24}\Gamma_3)^2)\\
&-c_3((M_{23}\Gamma_1+M_{12}\Gamma_{3}-M_{13}\Gamma_2)^2+(M_{24}\Gamma_1-M_{14}\Gamma_2+M_{12}\Gamma_4)^2)
\end{aligned}
\end{equation*}
So, we have
\begin{thm}\cite{DG3} The four dimensional Kirchhoff case is completely
integrable in the Liouville sense.
\end{thm}

In the case of the four-dimensional Chaplygin case, one can naturally
assume that $M_{12}$ and $M_{34}$ are  invariant relations. From
this assumption, we get:
\begin{dfn}\cite{DG3}
The four-dimensional Chaplygin case of the Kirchhoff equations on
$e(4)$ is defined by the Hamiltonian:
$$
\begin{aligned}
2H_{Ch}=&A_{1212}M_{12}^2+A_{1313}(M_{13}^2+M_{14}^2+M_{23}^2+M_{24}^2)+A_{3434}M_{34}^2+\\
&A_{1234}M_{12}M_{34}+A_{1213}M_{12}M_{13}+A_{1214}M_{12}M_{14}+\\
&A_{1223}M_{12}M_{23}+A_{1224}M_{12}M_{24}+A_{1334}M_{13}M_{34}+\\
&A_{1434}M_{14}M_{34}+A_{2334}M_{23}M_{34}+A_{2434}M_{24}M_{34}+\\
&B_{121}M_{12}\Gamma_1+B_{122}M_{12}\Gamma_2+B_{123}M_{12}\Gamma_3+B_{124}M_{12}\Gamma_4+\\
&B_{341}M_{34}\Gamma_1+B_{342}M_{34}\Gamma_2+B_{343}M_{34}\Gamma_3+B_{344}M_{34}\Gamma_4+\\
&C_{11}(\Gamma_1^2+\Gamma_2^2)+C_{33}(\Gamma_3^2+\Gamma_4^2).
\end{aligned}
$$
\end{dfn}

One can easily check that in this case $M_{12}$ and $M_{34}$ are
really the invariant relations.

\section*{Acknowledgments}

The research was partially supported by the Serbian Ministry of Education and
Science, Project 174020 Geometry and Topology of Manifolds,
Classical Mechanics and Integrable Dynamical Systems.
I would like to express my gratitude to Milena Radnovi\'c and Bo\v zidar Jovanovi\' c for helpful
remarks.
Also, I would like to thank the referee for useful comments and remarks.

\end{document}